\documentclass[showpacs,amsmath,amssymb,twocolumn,superscriptaddress,notitlepage,preprintnumbers,pra,aps, 10pt, floatfix]{revtex4-1}
\usepackage[dvips]{graphicx}
\usepackage{amsmath,amssymb,amsthm,mathrsfs,amsfonts,dsfont}
\usepackage{epsfig}
\usepackage{braket}
\usepackage{bm}
\usepackage{enumerate}
\usepackage{color}
\usepackage{qcircuit}
\usepackage[normalem]{ulem}

\usepackage{comment}
\usepackage{algorithm}
\usepackage{algpseudocode}
\usepackage{hyperref}
\hypersetup{colorlinks=true, linkcolor=blue, citecolor=blue, urlcolor=black}
\usepackage{xcolor}

\usepackage{tikz}
\usepackage{physics}
\usepackage[compat=0.6]{yquant}
\usetikzlibrary{quotes}
\usetikzlibrary{fit}
\useyquantlanguage{groups}
\usepackage{multirow}
\newcommand{\backvec}[1]{\reflectbox{$\vec{\reflectbox{\!$#1$}}$}}
  
\newcommand{\ch}[1]{\mathcal{#1}}
\newcommand{\bs}[1]{\boldsymbol{#1}}
\newcommand{\sch}[1]{\boldsymbol{\mathcal{#1}}}

\newcommand{\etal}{\textit{et al}.~}
\renewcommand{\dag}[1]{{#1}^{\dagger}}

\newcommand{\change}[1]{\textcolor{black}{#1}}
\newcommand{\new}[1]{\textcolor{black}{#1}}

\usepackage{varioref}
\labelformat{equation}{(#1)}

\labelformat{figure}{#1}

\labelformat{Section}{#1}

\labelformat{proposition}{#1}

\labelformat{lemma}{#1}

\labelformat{theorem}{#1}

\labelformat{corollary}{#1}

\labelformat{observation}{#1}

\labelformat{definition}{#1}

\newtheorem{theorem}{Theorem}
\newtheorem{proposition}{Proposition}

\begin{document}

\title{Purification and correction of quantum channels by commutation-derived\\ quantum filters}

\author{Sowmitra Das}
\email{sowmitra.das.sumit@gmail.com}
\affiliation{Blackett Laboratory, Imperial College London, London SW7 2AZ, United Kingdom}
\affiliation{School of Data and Sciences, BRAC University, Dhaka 1212, Bangladesh}

\author{Jinzhao Sun}
\email{jinzhao.sun.phys@gmail.com}
\affiliation{Blackett Laboratory, Imperial College London, London SW7 2AZ, United Kingdom}
\affiliation{School of Physical and Chemical Sciences, Queen Mary University of London, London, E1 4NS, United Kingdom}

\author{Michael Hanks}
\affiliation{Blackett Laboratory, Imperial College London, London SW7 2AZ, United Kingdom}

\author{B\'alint Koczor}
\affiliation{Mathematical Institute, University of Oxford, Oxford OX2 6GG, United Kingdom}

\author{M. S. Kim}
\affiliation{Blackett Laboratory, Imperial College London, London SW7 2AZ, United Kingdom} 

\begin{abstract}




Reducing the effect of errors is essential for reliable quantum computation. Quantum error mitigation (QEM) and quantum error correction (QEC)  are two leading frameworks for addressing this task, each with its respective challenges: sampling costs and inability to recover the state for QEM, and qubit and gate overheads for QEC. In this work, we combine ideas from these two frameworks and introduce an information-theoretic machinery called a quantum filter that can purify or correct quantum channels. We provide an explicit construction of the filter that can correct arbitrary types of noise in an $n$-qubit Clifford circuit using $2n$ ancillary qubits based on a commutation-derived error detection circuit. This filtering scheme can also partially purify noise in non-Clifford gates, such as T and CCZ gates. In contrast to QEC, this scheme works in an error-reduction sense because it does not require prior encoding of the input state into a QEC code and requires only a single instance of the target channel. Under the assumptions of clean ancillary qubits, this scheme overcomes the exponential sampling overhead in QEM because it can deterministically correct the error channel without discarding any result. The filter construction also provides a concrete protocol to implement a large class of quantum superchannels. We further propose an ancilla-efficient Pauli filter which proves to remove nearly all the low-weight erroneous Pauli components in a noisy Clifford circuit using only two ancillas. We prove that for local depolarising noise, this filter can achieve a quadratic reduction in the average infidelity of the channel. The Pauli filter can also be used to convert an unbiased error channel into a completely biased error channel and thus is compatible with biased-noise QEC codes which have high code capacity. \change{Beyond existing QEM works, it provides a systematic way to correct errors as the infidelity can be exponentially reduced with the number of ancillary qubits, using Clifford gates only.} \new{We numerically test the correction and purification capability of quantum filters under ancilla noise and identify the noise-rate regimes in which it outperforms other methods.}
Our results demonstrate the effectiveness of the quantum filter as an error-reduction technique in quantum information processing. 


\end{abstract}

\date{\today}

\maketitle

\section{Introduction}
\label{section:introduction}

A fundamental obstacle on the road to large-scale quantum computation is protecting quantum information from errors. The theoretical framework proposed for this purpose is quantum error correction (QEC). In QEC, the state of a single logical qubit is encoded into the larger Hilbert space of multiple physical qubits. The errors act on the larger space in such a way that appropriate measurements can determine the nature of the error and subsequent correction operations can revert the system back to its initial state \cite{lidar2013quantum, nielsen2010quantum}. 
However, the implementation of a fully error-corrected quantum computer imposes lofty requirements on the qubit number and the average fidelity of each gate operation. Quantum error mitigation (QEM) has been proposed as an intermediate alternative to QEC, which mitigates errors in expectation values by post-processing the results of noisy circuit-runs~\cite{temme_error_2017, endo_practical_2018,sun2021mitigating,huggins_virtual_2021,koczor_exponential_2021,bonet2018low, mcardle2019error,yoshioka_generalized_2022,strikis_learning-based_2021, czarnik2021error,cai_quantum_2022,endo_hybrid_2021}. As a consequence, QEM cannot recover the noiseless state, in contrast to QEC.
More importantly, investigations invoking the data-processing inequality have shown that exponentially many samples (in circuit width and depth) are required to mitigate errors under typical noise models \cite{takagi2022fundamental, takagi2023universal, tsubouchi2023universal, quek2024exponentially}, thus defeating the advantage of quantum computing in the first place. This indicates that purely relying on post-processing with no additional qubits would not be adequate to scale these methods up. We would need to use controlled amounts of quantum resources and intermediate forms of error-correction to efficiently achieve error suppression without an exponential sampling overhead, which is the major objective of this work.


\new{Since a primary drawback of the first-generation QEM protocols is their exponential sampling overhead, a number of recent works have focused on circumventing this limitation.
One strategy is to combine the ideas of QEM and QEC \cite{lostaglio2021error, piveteau2021error,suzuki2022quantum}, mostly focusing on gate purification using probabilistic error cancellation (PEC). 
However, since these methods apply QEM to a sub-routine of the QEC pipeline without changing the underlying structure of the QEM protocol, they still suffer from the exponential sampling overhead.
Another route to overcome the limitation is to use additional quantum resources, in the form of clean ancillas. The use of ancillas is motivated by the fact that in the transition from noisy intermediate-scale quantum (NISQ) to early fault-tolerant devices, inhomogeneities in fabrication, environment and control latency naturally give rise to asymmetries in noise at both the physical and  logical levels. Such hardware enables us to create and couple qubits with vastly different error-rates \cite{araki2025space, strikis2023quantum, bultrini2023battle, koukoulekidis2023framework, breuckmann2017hyperbolic}, some of which can be considered much cleaner than others. Also, different physical platforms have different advantages and disadvantages in scalability, noise levels and gate speeds. There have been discussions to combine different physical platforms for the best use of their respective advantages, for instance, a hybrid between a superconducting system for scalability and fast gate operations and an ion system for cleaner gate operations.
Meanwhile, in the early fault-tolerant regime, assumptions about noise homogeneity break down not only at the physical level, but also at the logical level. The majority of logical qubits have limited code distance and remain susceptible to such errors, while a small number of higher-distance logical qubits serve as relatively cleaner ancillas.} 

\new{A promising class of protocols that take advantage of distributed and inhomogeneous quantum architectures can be broadly classified as Ancilla-assisted Quantum Error Mitigation (AQEM). Compared to earlier versions of QEM, they make use of clean ancilla qubits to coherently control multiple noisy instances of a target state or channel, which, if the protocol is successful, results in a less noisy version of the target. Protocols pertaining to this model are virtual channel purification \cite{liu2024virtual}, virtual state purification \cite{koczor_exponential_2021, childs2023streaming,huggins_virtual_2021}, superposed QEM \cite{miguel2023superposed}, and black-box error suppression \cite{lee2023error}. Nonetheless, these proposals do not address the fundamental cost issue and suffer from their own drawbacks. Firstly, they operate in virtual or post-selection mode, requiring multiple runs to succeed or obtain enough samples for averaging. They offer no mechanism to \emph{correct} for the errors detected in individual runs. Secondly, all these protocols exclusively use the (multiple-)controlled-\textsc{swap} gate to perform the purification, which is a multi-qubit non-Clifford gate and thus hard to implement. 
Thirdly, they require coherent queries to multiple instances (at least two instances \cite{miguel2023superposed,lee2023error,liu2024virtual,yang2024quantum}) of the noisy circuit, with the amount of purification achieved depending on the number of queries, which is challenging for coherently storing and controlling.}




\new{In this work, we design an error-reduction protocol that overcomes the exponential sampling overhead of QEM by providing a way to correct a noisy channel using a single query to the channel. Our approach relies solely on simple Clifford operations and does not require explicitly encoding the input state in a QEC code. Under the assumption of perfect ancillary qubits, this exponentially improves the output fidelity of the channel with each added ancilla. In doing so, we overcome the challenges set up by the limitations of QEM and the intermediate forms between QEM and QEC mentioned above.}

\textbf{Overview of main results.}
We introduce an information-theoretic machinery called a {quantum filter}. {Formally, a quantum filter provides a general way to implement a class of quantum superchannels aimed at removing sub-components of a target channel.} As an application of this formalism, we propose a method to detect or correct errors using a single query to the circuit and simple controlled operations (specifically, controlled Pauli operations). We first introduce the idea of the {commutation-derived filter}, which can probe a channel by exploiting the commutation structure of the channel components, and {filter out} unwanted components. Then, we extend this construction to correct the noise component instead of just discarding it, so that we can deterministically recover the noiseless state. To do so, we first formulate the correction scheme for a stochastic Pauli channel, and extend it to purify arbitrary quantum channels given the prior information about the commutation structure. 

\change{This formalism is then generalised by drawing inspiration from the linear-combination-of-unitaries (LCU) circuits~\cite{childs2012hamiltonian,sun2024high}. As an application of the generalization, we construct a filter that can completely purify Clifford channels. We also show that the filter can partially purify non-Clifford gates, such as the T-gate and the CCZ-gate, towards $Z$-biased noise. In general, even in scenarios where we cannot correct the entire channel, we show that we can bias the errors towards a certain axis (either $X$, $Y$ or $Z$) using a Pauli filter. \new{Therefore, this protocol can restore the bias in a circuit composed with biased-noise in mind, e.g, biased-noise QEC codes which usually have high code capacity \cite{tuckett2018ultrahigh,tuckett2019tailoring,xu2023tailored}.} In addition, if the noise is fully biased, our scheme could completely remove the noise in general quantum circuits. The correction capability of the filter allows us to exponentially increase the output fidelity with the number of ancillas with at most $2n$ clean qubits. 
Thus, the filter overcomes the exponential sampling overhead of QEM, and does this using a single query to the noisy circuit with only Clifford resources - features that no contemporary QEM work~\cite{lee2023error,yang2024quantum, liu2024virtual, miguel2023superposed} can achieve simultaneously.  Although all these recent QEM schemes use some number of extra qubits, the question of how to best use these extra qubits in general and how the output accuracy scales with them has largely remained unexplored. Our work thus addresses this gap by providing a scalable and tunable strategy to reduce errors based on the number of extra qubits that are available.}


For near-term applications, it is desirable to introduce fewer resources than those required to correct the entire channel. We provide the construction of an ancilla-efficient Pauli filter, which operates in the error-detection mode using only two ancillas and can remove all the weight-1 Pauli errors in a noisy Clifford channel. 
{To remove $k$-local noise, the number of ancillas is shown to be lower bounded by $k$.}
We show that this ancilla-efficient Pauli filter can additionally remove a large portion of the other low-weight Pauli components of the error channel. Specifically, the average weight of all Paulis that can be removed is $n/2$ asymptotically (Theorem 2), which gives a theoretical guarantee for the filtration of low-weight Pauli errors.
We prove that for a local depolarising noise model, for example, the ancilla-efficient  Pauli filter can achieve a quadratic reduction in the infidelity of the channel (Theorem 1). This exceeds the result in \cite{lee2023error} as it requires only a single query to the original circuit, albeit for the restricted class of circuits and noise. \new{Comparison of the qubit/gate overheads of our method can be found in \autoref{table:resource-comparison}, which shows improvements compared to state-of-the-art AQEM protocols.}



\new{We validate the theoretical analysis of the correction filter and the ancilla-efficient filter with numerical results, where ancillary qubit noise is considered as well. We also elucidate similarities and key differences with small-scale QEC methods, including  flag QEC codes~\cite{chao2018quantum,chao2020flag,chamberland2018flag}
and the Iceberg code~\cite{self2024protecting}. As an example, we compare our method with flag QEC, and find that our method shows advantages in moderate to high gate-noise regimes. We conclude with discussions about the cost of implementing our method practically, analysis about the assumptions made in our protocol, and an outline of possible hardware architectures in which our protocol may be implemented.}

Our results demonstrate the versatility of the quantum filter as an error-reduction method: one can use it either in error detection or correction mode, tune the amount of purification based on the number of ancillary qubits available, select which gates to purify and which noise components to remove. Thus, flexibility in configuring a quantum filter's capabilities is shown depending on the amount of resources at hand.

\section{Quantum Channel Filtration}
\label{section:quantum-channel-filtration}

 

Here we introduce the concept of the quantum filter, restricting our examples to simple storage/communication channels. The main idea of quantum channel filtration is that we can probe an erroneous quantum channel by using an error-free operation that obeys a certain commutation relation with respect to the ideal channel. The probing operation, if chosen properly,  separates the error in the original channel into sub-channels.
Depending on the measurement outcome, the computation is then either aborted or continued.
This is reminiscent of filtering communication channels by their frequency components in classical communication theory, which motivates us to name this scheme a {quantum filter}
({this is different from a similarly named scheme introduced in \cite{lee2023error}). We show the basic principles of this construction in \autoref{sec:commutation_filter}}.

As a single probe may not be sufficient to filter out all the errors, we show how to
successively concatenate multiple probes to achieve higher purification of the input channel in
\autoref{section:successive-filtration}. Finally, we show in \autoref{sec:quantum-channel-correction} that
instead of discarding the erroneous channels, it is possible to correct the errors based on the outcomes
of the probes.  Given the present formalism depends only on the quantum channel to be purified---and has no
dependence on the input state---we refer to this framework as {quantum channel filtration} and {quantum channel correction}. 

\subsection{The commutation-derived quantum filter}
\label{sec:commutation_filter}

{First, we develop the machinery of commutation structures that we are going to exploit to detect and filter out errors in a quantum channel.}\\

\paragraph{Commutation-derived error detection circuit}

Suppose our aim is to ideally implement a unitary $U$ but in a way that we can detect if certain errors have occurred during execution (but without prying open $U$, or {measuring} the output state).
To do this, suppose we have access to a unitary $S$ which commutes with $U$, i.e., $[S, U] = 0$.
We then use the following circuit to probe $U$ using $S$ as
\begin{center}
\begin{tikzpicture}
\begin{yquant}
qubit {$\ket{0}$} q;
qubit {$\ket{\psi}$} r;
h q;
box {$S$} r|q;
[shape=yquant-rectangle, rounded corners=.3em, fill=red!20]box {$U$} r; 
box {$S^{\dagger}$} r|q ; 
h q;
[label=$\ket{{\tt out}}$]barrier (q, r);
measure q;
discard q; 
\end{yquant}
\end{tikzpicture}
\end{center}
This type of circuit is used for error-detection in coherent parity check codes~\cite{roffe2018protecting,van2023single}. It is straightforward to see that the output state $\ket{\tt out}$ of the circuit (before the measurement) is:
\begin{equation}\label{eqn:vanilla_filter_output}
    \ket{\tt out} = \frac{1}{2}\ket{0}\otimes (U + S^{\dagger}U S)\ket{\psi} + \frac{1}{2}\ket{1}\otimes (U - S^{\dagger}U S)\ket{\psi}
\end{equation}
If $U$ commutes with $S$, the output state will be
\begin{equation}
    \ket{\tt out} = \ket{0}\otimes U\ket{\psi}.
\end{equation}
Hence,  we always get the outcome $\ket{0}$ when measuring the ancilla. If an outcome of $\ket{1}$ is incidentally obtained, it is certain that an error has occurred in $U$ and we can discard the circuit-run to avoid the error.
\\

Upon taking a closer look at  \autoref{eqn:vanilla_filter_output}, we can see that, based on the measurement outcome of the ancilla,
the qubits in the state $\ket{\psi}$ collapse to the following states as
\begin{align*}
    \ket{0}: (U + \dag{S}US)\ket{\psi}, \qquad  
    \ket{1}: (U - \dag{S}US)\ket{\psi}
\end{align*}

Now, if we assume that the unitary $U$ has the property that, it either commutes or anti-commutes with another unitary $P$, 
$ 
    PU = \pm UP
$
then, if we plug in $P$ in place of $S$, the output of the circuit is
\begin{equation}
    \ket{\tt out} = \begin{cases}
        \ket{0}\otimes U\ket{\psi} & \text{if } U\; \text{commutes with}\; P\\
        \ket{1}\otimes U\ket{\psi} & \text{if } U\; \text{anti-commutes with}\; P
    \end{cases}
\end{equation}
Therefore, the measurement result on the ancilla reveals the commutation structure of $U$
with respect to $P$, which can be used to detect errors in the channel. \\ 

\paragraph{Filtering Quantum Channels} Now, we generalize the {above error detection circuit further to a quantum filter}. Consider the stochastic channel,
$ \mathcal{E} = \sum_i p_i  \mathcal{U}_i$
where $p_i$ are probabilities and $\mathcal{U}_i$ are unitary channels
that act by conjugation as $\mathcal{U}_i(\rho) = U_i\rho U_i^{\dagger}$. Applying our error detection circuit 
in to this channel yields
\begin{center}
\begin{tikzpicture}
\begin{yquant}
qubit {$\ket{0}$} q;
qubit {} r;
h q;
box {$P$} r|q;
[shape=yquant-rectangle, rounded corners=.3em, fill=red!20]box {$\mathcal{E}$} r; 
box {$P^{\dagger}$} r|q ; 
h q;
measure q;
discard q; 
\end{yquant}
\end{tikzpicture}
\end{center}

Suppose that there exists a $P$, such that it either commutes or anti-commutes with the unitaries $U_i$
as
\begin{equation}
    PU_i = \pm U_iP \qquad {\rm for\; all\;} i.
\end{equation}
This condition divides the components of $\mathcal{E}$ into 2 disjoint subsets.
Let us denote these components as $\{U_j^+\}$ for those commuting with $P$ and $\{U_k^-\}$ for those anti-commuting with $P$  with corresponding probabilities $\{p_j^+\}$ and $\{p_k^-\}$. Then, conditioned on the measurement result, the channel components $U_j^+$ and $U_k^-$ are transformed to
\begin{align*}
    \ket{0} &: U_j^+ \mapsto U_j^+, \qquad U_k^- \mapsto 0\\
    \ket{1} &: U_j^+ \mapsto 0, \quad\qquad U_k^- \mapsto U_k^-
\end{align*}
Thus, the entire channel $\mathcal{E}$ is transformed to 
\begin{align}
    \ket{0} : \mathcal{E} \mapsto  \mathcal{E^+} = \frac{1}{p^+}\sum_j p_j^+ \cdot \mathcal{U}_j^+  \nonumber\\
    \ket{1} : \mathcal{E} \mapsto  \mathcal{E^-} = \frac{1}{p^-}\sum_k p_k^- \cdot \mathcal{U}_k^- \nonumber
\end{align}
where $p^+ = \sum_j p_j^+$, and $p^- = \sum_k p_k^-$, and the probability of the two outcomes is denoted as $ 
    p(\ket{0}) = p^+$ and $p(\ket{1}) = p^-
$.
As such, post-selecting on the outcome $\ket{0}$  ($\ket{1}$) of the ancilla qubit
removes all components of the channel $\mathcal{E}$ that anti-commute (commute) with $P$.
Given this circuit selects components of a channel based on their commutation structure with respect to a unitary $P$,
we will refer to it as a {commutation filter}. More generally, we call any circuit which can filter out certain components of a noise channel a quantum filter. 
Furthermore, as the present quantum filter maps one quantum channel ($\mathcal{E}$) to another ($ \mathcal{E^+}$ or $ \mathcal{E^-}$), it is formally a quantum superchannel \cite{chiribella2008transforming}. 
{A description of quantum filters using the language superchannels is discussed in \autoref{section:supermap}.}\\

\paragraph{Generalization} The commutation filter can be generalized to arbitrary channels $\mathcal{E}$ and more general operators $P$. We first provide the following proposition. 

\begin{proposition}
\label{theorem:commutation}
    For an arbitrary operator $A$ and unitary operator $F$ such that, $F^2$ commutes with $A$,
    $ 
        [F^2, A] = 0
    $
    we can write $A$ as the sum of two components, 
    $ 
        A = A_+ + A_-
    $
    such that, 
    $ 
        FA_{\pm} = \pm A_{\pm}F$
    i.e, one of them commutes with $F$ and the other anti-commutes with $F$. 
\end{proposition}

\begin{proof}
    We provide a proof by construction. Consider the operators, 
    \begin{equation}
        A_{\pm} := (A\pm FAF^{\dagger})/2
    \end{equation}
    {It is easy to check that $A_+ + A_- = A$. Plugging these expressions into $FA_{\pm}$, and using $F^2A = AF^2$, one can verify the equalities $FA_{\pm} = \pm A_{\pm}F$, which concludes our proof.}
\end{proof}
 
Let us assume an arbitrary channel $\ch{E}$ in terms of its Kraus operators $E_i$
as $ \ch{E}(\rho) = \sum_iE_i\rho E_i^{\dagger}$, such that an operator $F$ exists
whose square $F^2$ commutes with all Kraus operators $E_i$ as 
\begin{equation}\label{eqn:kraus_commutation}
    [F^2, E_i] = 0 \qquad \forall \, i.
\end{equation}
Due to the above proposition, it then follows that
all Kraus operators $E_i$ indeed decompose as a sum of commuting and anti-commuting components as
$ 
    E_i = E_{i+} + E_{i-}.
$
Plugging the channel $\ch{E}$ into a commutation filter via $P = F$, 
the filter selects either the commuting or anti-commuting components
depending on the outcome of the ancilla measurement as
$
    \ket{0}: \ch{E} \mapsto \ch{E}^+ = \frac{1}{e^+}\sum_i E_{i+}(\cdot)E_{i+}^{\dagger}$, $
    \ket{1}: \ch{E} \mapsto \ch{E}^- = \frac{1}{e^-}\sum_i E_{i-}(\cdot)E_{i-}^{\dagger}
$
where $e^{\pm}$ are normalization factors, 
$
    e^{\pm} = \sum_i \Tr[E_{i\pm}(\cdot)E_{i\pm}^{\dagger}].
 $
A subset of all admissible operators $F$ are involution operators
whose square is by definition the identity operator as $F^2 = I$.
Well-known examples of such operators are Pauli strings and \textsc{swap} operators, which has been discussed in \cite{liu2024virtual} in a similar context. For involution operators, \autoref{eqn:kraus_commutation} holds trivially.

\subsection{Channel purification by successive filtration}
\label{section:successive-filtration}

{We have shown that given an appropriately chosen probing operator $P$, the commutation filter can remove certain erroneous components of an input channel. Usually, a single filtration step will not be able to remove all the errors. Here, we show how to concatenate successive filtration steps to potentially remove all the error components in the channel.}

We return to our first example whereby we aim to implement a unitary channel 
$\ch{U}_0(\rho) = U_0\rho U_0^{\dagger} $
and assume a noise model such that the actual channel is $\ch{U_E}(\rho) = \sum_{i=0}^m p_i \cdot \ch{U}_i$.
Here the index value $i \geq 1$ indicates unitary $U_i$ as the noise-component.
Now, if we can find an operator $F$ which commutes with $U_0$ and anti-commutes with all other $U_i$'s, we can filter out $U_0$ by plugging $\ch{E}$ into a commutation filter with $P=F$. The probability of successful filtration will be $p_s = p_0$.

In general, a single $F$ which commutes with $U_0$ and anti-commutes with all other $U_i$ might not exist, or might be difficult to construct. 
{Nonetheless, this requirement could be relaxed. For example, suppose that we can construct a set of multiple operators $\{F_i\}$ such that
only $U_0$ commutes with all of them, and all the other $U_i$'s anti-commute with at least one $F_i$.}
Then, we can purify the channel $\ch{U_E}(\rho)$ by successively filtering it using $F_{i \geq 1}$.
If all the rounds are successful, the component $U_0$ will survive all the filters, and all other components
will be removed by at least one filter. The probability that all the rounds are successful is still $p_0$. 

\paragraph{Noiseless Channel from a Noisy Pauli Channel}
\label{exm:single-qubit-example}
 
Consider the (single-qubit) stochastic Pauli-channel
which would ideally ($p_I = 1$) implement the identity operation  as
$ 
    \ch{E}(\rho) = p_I\cdot\rho + p_X\cdot X\rho X + p_Y\cdot Y\rho Y + p_Z\cdot Z\rho Z,
$
where the channel components $I, X, Y$ and  $Z$ are Pauli operators and $p_I + p_X + p_Y + p_Z = 1$.
Consider the commutation filters $\sch{F}_Z, \sch{F}_X$ using the Pauli operators $P \in \{Z, X\}$ implemented one after the other 
as in \autoref{fig:successive-filtration}.
If we filter $\ch{E}$ first using $Z$, and then using $X$ (post-selecting on the measurement 0 each time), the components that survive both steps are
\begin{equation}
    \{I, X, Y, Z\} \xrightarrow[\ket{0}]{\quad\sch{F}_Z\quad} \{I, Z\} \xrightarrow[\ket{0}]{\quad\sch{F}_X\quad} \{I\}.
\end{equation}

The probability of success of the first round is  $ p_{s1} = p_I {+} p_Z$
while the probability of success of the second round, given that the first round was successful, is 
$p_{s2|s1} = {p_I}/(p_I + p_Z)$.
Indeed, the probability that the entire filtration protocol is successful is given by 
\begin{equation}
    p_s = p_{s1} \cdot p_{s1|s2} = p_I.
\end{equation}


\begin{figure}
    \begin{center}
    \begin{tikzpicture}[scale=0.9]
    \begin{yquant}
    qubit {$\ket{0}$} q1;
    qubit {} q2;
    discard q2; 
    qubit {} r;
    h q1;
    box {$X$} r|q1;
    
    init {$\ket{0}$} q2; 
    h q2; 
    box {$Z$} r|q2; 
    [shape=yquant-rectangle, rounded corners=.3em, fill=red!20]box {$\mathcal{E}$} r; 
    
    box {$Z$} r|q2 ; 
    h q2;
    measure q2;
    init {$\bra{0}$} q2; 
    discard q2; 
    
    box {$X$} r|q1;
    h q1; 
    measure q1;
    output {$\bra{0}$} q1;
    \end{yquant}
    \end{tikzpicture}
    \end{center}
    \caption{Successive filtration scheme for purifying a single-qubit error channel.}
    \label{fig:successive-filtration}
\end{figure}
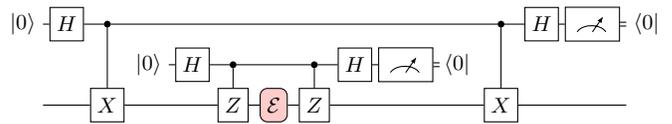

As can be seen, we can probabilistically purify a single-qubit noisy Pauli channel using a $Z$-Filter and an $X$-Filter in sequence. This can be easily generalized to an arbitrary number of qubits. To purify an $n$-qubit Pauli channel using Successive Filtration, we need $2n$ ancillary qubits, $n$ $Z$-Filters and $n$ $X$-Filters. 

\subsection{Quantum channel correction}\label{sec:quantum-channel-correction}

{In the previous section, we have shown that a channel can be purified by successively discarding erroneous components at each filtration step. The part of the channel that survives all the steps has all the corresponding errors removed. Whether a certain step will be successful depends on the probabilistic outcome of the probe. Here we show that, by slight modification of the above scheme, instead of probabilistically purifying a quantum channel, we can deterministically {correct} the channel. This stems from the observation that if a round is unsuccessful, we can still use the information
	from the measurement outcomes to correct for the corresponding error components. }

Suppose that in example \ref{exm:single-qubit-example}, the first round with the $Z$-filter is unsuccessful, i.e, the outcome of the first ancilla is $\ket{1}$. Then, the components of the channel which anti-commute with $Z$ survive resulting in the mapping
\begin{equation}
    \{I, X, Y, Z\} \xrightarrow[\ket{1}]{\quad \sch{F}_Z\quad} \{X, Y\}
\end{equation}
In this case, the channel $\ch{E}$ above is transformed to 
\begin{equation}
    \ch{E}_1= \frac{1}{p_X + p_Y}\Big(p_X X(\cdot)X + p_Y Y(\cdot)Y\Big),
\end{equation}
and thus applying the Pauli operator $X$ converts all components of this channel to ones that commute with $Z$
as
$ 
    X\cdot\{X, Y\} \rightarrow \{I, Z\},
$
thereby {deterministically} converting all channel components to $\{I, Z\}$. Now, if we repeat this procedure for the $X$-filter, and apply the $Z$ gate conditioned on the ancilla of this filter returning 1, we can deterministically convert both components to $I$. As a result, at the end of both rounds, the channel will be {corrected} to the identity channel. 
The circuit that implements this scheme of {quantum channel correction} is shown in \autoref{fig:channel-correction}.

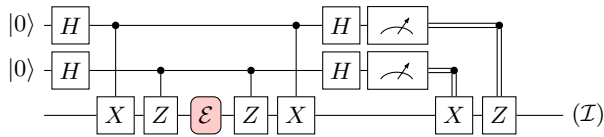
\begin{figure}[tb]
    \begin{center}
    \begin{tikzpicture}
    \begin{yquant}
    qubit {$\ket{0}$} q1;
    qubit {$\ket{0}$} q2;
    qubit {} r;
    h q1, q2;
    box {$X$} r|q1; 
    box {$Z$} r|q2; 
    [shape=yquant-rectangle, rounded corners=.3em, fill=red!20]box {$\mathcal{E}$} r; 
    
    box {$Z$} r|q2 ;  
    box {$X$} r|q1;
    h q1, q2;
    measure q1;
    measure q2;
    
    box {$X$} r|q2;
    box {$Z$} r|q1;
    discard q1; 
    discard q2;
    
    hspace {0.5cm} r;
    output {$(\ch{I})$} r;
    \end{yquant}
    \end{tikzpicture}
    \end{center}
    \caption{An example of a Quantum Filter for channel correction.}
    \label{fig:channel-correction}
\end{figure}

Although we have outlined the scheme for Pauli errors on a single-qubit, this is sufficient to correct any arbitrary errors on a qubit by a standard result in the theory of quantum error correction \cite{shor1995scheme}. We provide proof of this result for the concatenated Pauli correction filter of \autoref{fig:channel-correction} in Appendix \ref{app:proof-channel-correction}. 
A consequence of this result is that, although many QEM schemes are limited to correcting stochastic errors \cite{liu2024virtual, lee2023error} or known types of noise (which is required by PEC \cite{temme_error_2017}), our quantum channel correction scheme works for stochastic Pauli errors {and more general error models that may include coherent errors}. 
As our result applies to arbitrary single-qubit error channel correction, it can be easily seen that any $n$-qubit error channel can be corrected by running this scheme in parallel on each qubit using $2n$ ancillary qubits. 

\section{Construction of a quantum filter}

In this section, we extend the description beyond simple communication channels, providing an explicit construction to correct for noise in arbitrary Clifford circuits. To do so, we first develop a more general construction of the filter circuit. We formulate the description of a quantum filter in the language of quantum superchannels.

Briefly, a quantum superchannel $\sch{F}$ (written in bold calligraphic) is a map that takes quantum channels to quantum channels. A superchannel $\sch{F}$ is described by a set of {super-Kraus} elements $\{\bs{F}_j\}_j$, such that, if $\sch{F}$ acts on a quantum channel $\ch{E}$ with Kraus operators $\{E_k\}_k$, the Kraus operators of the transformed channel $\sch{F}[\ch{E}]$ are 
\begin{equation}
    \sch{F}[\ch{E}] : \{E_k\}_k \longrightarrow \{\bs{F}_j[E_k]\}_{j, k}
\end{equation}
Each of the $\bs{F}_j$'s are linear maps on the space of operators $E_k$. If $\ch{E}$ acts on a state $\rho$ as 
$ 
    \ch{E}(\rho) = \sum_k E_k\rho E_k^{\dagger},
$
the transformed channel $\sch{F}[\ch{E}]$ acts on $\rho$ according to
\begin{equation}\label{eq:superchannel}
    \sch{F}[\ch{E}](\rho) = \sum_j \sum_k \bs{F}_j[E_k] \rho \bs{F}_j[E_k]^{\dagger}.
\end{equation}
If the set $\{\bs{F}_j\}$ contains a single element, we call $\sch{F}$ a {unitary} superchannel. Readers interested in this formalism may consult \cite{chiribella2008transforming, chiribella2008quantum, daly2022axiomatic} for further details. 

\new{In the sections below, we cast the operation of our quantum filter in this superchannel formalism. This is a natural setting for our scheme, as the quantum filter directly targets a noise channel, rather than the input state per se. While \autoref{eq:superchannel} denotes the most general way in which a superchannel can be expressed, the filter construction represents a special class of superchannels that can be used to remove error components from a target-channel $\ch{E}$, by manipulating their commutation structure.}

\subsection{Quantum filters as superchannels}
\label{section:supermap}
First, we focus on the circuit of the basic quantum filter shown in \autoref{fig:supermap}.




\begin{figure}[!h]
    \centering
    \includegraphics[width=0.75\columnwidth]{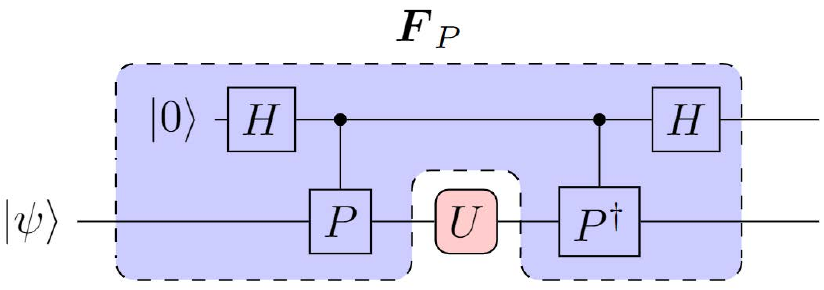}
    \caption{Quantum filter outlining the filter supermap. The output is given by \autoref{eqn:vanilla_filter_output}.
   { If  the filtering
qubits are discarded, this will result in a non-unitary quantum filter. We denote the corresponding superchannel by $\widetilde{\sch{F}}$.}
    }
    \label{fig:supermap}
\end{figure}
 
Consider the part of the circuit outlined in blue, which we denote by $\bs{F}_{P}$. This can be thought of as a circuit element that takes the unitary $U$ as input, and produces an isometry that takes the system qubits to the combined system-ancilla state. This map is described as
\begin{align}
    &\boldsymbol{F}_{P}[U]: \nonumber\\
    &U \mapsto \frac{1}{2}\ket{0}\otimes (U + \dag{P}UP) + \frac{1}{2}\ket{1}\otimes (U - \dag{P}UP)\label{eqn:filter-supermap}
\end{align}
In this form, instead of thinking about how the input state is transformed, we focus on how the filter transforms the unitary $U$ (or the channel $\ch{U} = U(\cdot)U^{\dagger}$). We call the associated map $\sch{F}_P$, that transforms the channel $\ch{U}$ to $\sch{F}_p[\ch{U}]$, the {Filter Superchannel} \cite{chiribella2008transforming}.
\begin{equation}\label{eqn:unitary-filter}
    \sch{F}_p[\ch{U}](\rho) = \bs{F}_p[U]\rho\bs{F}_p[U]^{\dagger}
\end{equation}
Since $\sch{F}_{P}$ is defined by a single map $\bs{F}_P$, it is a unitary supermap \cite{daly2022axiomatic}. In this language, $\bs{F}_P$ is a super-Kraus element of the superchannel $\sch{F}_P$.

\subsubsection{Symmetric (Unitary) Pauli Filters}
We illustrate the above formulation by analyzing the (single-qubit) Pauli $Z$ and Pauli $X$ filters, i.e, $P~\in~\{Z, X\}$. We denote these filters simply by $\sch{F}_Z, \sch{F}_X$ respectively. For convenience, we define $(I, Z, X, Y) = (P_{00}, P_{01}, P_{10}, P_{11})$. 
Using \autoref{eqn:filter-supermap}, we can show that, 
\begin{align}
    \bs{F}_{Z}[P_u] = \ket{u_0}\otimes P_u \nonumber\\
    \bs{F}_X[P_u] = \ket{u_1}\otimes P_u \nonumber
\end{align}
where $u=u_0u_1\in \{00, 01, 10, 11\}$. 
Successive filtration by these two filters (\autoref{fig:successive-filtration}) is equivalent to concatenating $\sch{F}_Z$ and $\sch{F}_X$. The resulting super-Kraus element $\bs{F}_{Z;X}$ is defined (on the Paulis) as 
\begin{equation*}
    \bs{F}_{Z;X}[P_u] = \bs{F}_Z[\bs{F}_X[P_u]] =  \ket{u_0}\otimes \ket{u_1}\otimes P_u = \ket{u} \otimes P_u
\end{equation*}
Now, if the Pauli Transfer Matrix $\chi_{u, v}$ of a (single-qubit) quantum channel $\ch{E}$ is given, such that
\begin{equation*}
    \ch{E}(\rho) = \sum_{u, v} \chi_{u, v}P_u\rho P_v, 
\end{equation*}
by applying the filter $\sch{F}_{Z;X}$ on this channel, we have, 
\begin{align}\nonumber
    \sch{F}_{Z;X}[\ch{E}](\rho) &= \sum_{u, v}\chi_{u, v}\bs{F}_{Z;X}[P_u]\rho \bs{F}_{Z;X}[P_v]^{\dagger}\\
    &= \sum_{u, v}\chi_{u, v} \ketbra{u}{v} \otimes P_u\rho P_v \label{eqn:tomography}
\end{align}
Hence, for any stochastic Pauli channel of the form 
$
    \ch{E} = p_I \ch{I} + p_Z\ch{Z} + p_X\ch{X} + p_Y\ch{Y} 
$, where $\ch{X}(\rho) = X\rho X$, etc. and for which the matrix $\chi_{u, v}$ is diagonal, we have, 
\begin{align}
    \sch{F}_{Z;X}[\ch{E}] =\, & p_I \dyad{00}\otimes \ch{I} + p_Z\dyad{01}\otimes\ch{Z} \nonumber\\
    + & p_Y\dyad{11}\otimes \ch{Y} + p_X\dyad{10}\otimes \ch{X} \nonumber
\end{align}
Thus, the ancilla now carries information about the channel components, which, as we have seen, can be used to filter them out or perform further post-processing.
Interestingly, \autoref{eqn:tomography} suggests that the coefficients $\chi_{u, v}$ might be found out by performing state-tomography on the ancillary qubits. This would be equivalent to performing process tomography of the channel $\ch{E}$. 
The quantum filter formalism can be directly applied to process tomography, and we leave detailed discussions to dedicated readers.


\subsubsection{Non-Unitary Filters}

In the above treatment, we considered only unitary operators as part of the filter circuitry. Here, we extend this treatment further to include non-unitary operations such as measurements and conditional operations as well. Suppose we choose to discard the filtering qubits from the circuit in \autoref{fig:supermap}, 
and denote the resulting superchannel as $\widetilde{\sch{F}}$. The action of $\widetilde{\sch{F}}$ can be obtained by tracing out the ancilla in \autoref{eqn:unitary-filter},  
\begin{align}
    \widetilde{\sch{F}}[\ch{U}](\rho) &= \bra{0}\bs{F}[U]\rho \bs{F}[U]^{\dagger}\ket{0} + \bra{1}\bs{F}[U]\rho \bs{F}[U]^{\dagger}\ket{1}\nonumber\\
    &= \bs{F_1}[U]\rho \bs{F_1}[U]^{\dagger} + \bs{F_2}[U]\rho \bs{F_2}[U]^{\dagger}\label{eqn:pq-superkraus}
\end{align}
where in the first line, the inner product is over the ancillary qubits, and in the second line, we define 
$ 
    \bs{F_1}[U] = \frac{1}{2}IUI + \frac{1}{2}P^{\dagger}UP$, $
    \bs{F_2}[U] = \frac{1}{2}IUI - \frac{1}{2}P^{\dagger}UP
$. The maps $\bs{F_1}, \bs{F_2}$ are super-Kraus elements of the filter $\widetilde{\sch{F}}$.

Thus, we have translated the Quantum Filter formalism in terms of the theory of quantum superchannels, which gives a natural language for our scheme. In Appendix \ref{app:channel-correction-superchannel}, we give an analysis of the channel-correction filter using this formalism, which may be used to provide an alternate proof of the channel correction property mentioned at the end of Section \ref{sec:quantum-channel-correction}. 

\subsection{Beyond symmetry: general construction of filters}
\label{section:general-quantum-filter}
Now, we extend the construction of a symmetric commutation filter to more general quantum filters. First, we observe that our original  Error-Detector (conditioned on the output of the ancilla being 0) maps $U$ to 
$$
  U \mapsto \frac{1}{2}U + \frac{1}{2}\dag{P}UP.
$$
This can be seen as a convex combination of $IUI$ and $\dag{P}UP$ (with equal probabilities). The equal coefficients come from the use of the Hadamard gate, and, $I$, $P$ and $\dag{P}$ come from the Controlled-$P$ and Controlled-$\dag{P}$ gates. 
In general, we might think of mapping $U$ to a general convex combination of the form
\begin{equation}\label{eq:filter-superchannel}
    U \mapsto \sum_i p_i\, Q_iUP_i
\end{equation}
for unitaries $P_i, Q_i$ and probabilities $p_i$. 
This transformation can be achieved by using the circuit in \autoref{fig:general-filter}, whose construction is inspired from the LCU circuits proposed by Childs and Wiebe \cite{childs2012hamiltonian}.   
\begin{figure}[!h]
    \centering
    \includegraphics[width=1\columnwidth]{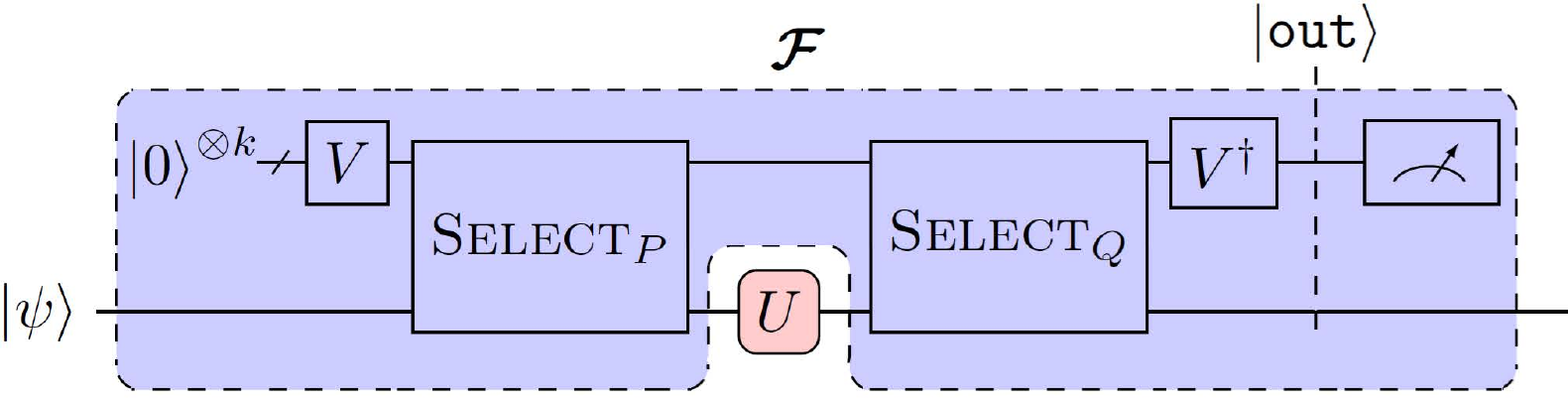}
    \caption{A schematic of a quantum filter specified by the prepare unitary V and the controlled unitaries $\textsc{Select}_P$ and $\textsc{Select}_Q$.
    }
    \label{fig:general-filter}
\end{figure}

This quantum filter $\sch{F}$  is specified by the {preparation} unitary $V$ and the controlled unitaries $\textsc{Select}_P$ and $\textsc{Select}_Q$. The ancilla can be in general a multi-qubit register, which we call the \emph{Filtering Qubits}. The probabilities $p_i$ are encoded in the unitary $V$, which is defined as 
$ 
    V\ket{0} = \sum_i \sqrt{p_i}\ket{i}
$
and, the unitary operators $P_i$ and $Q_i$ are encoded in the controlled unitaries $\textsc{Select}_P$ and $\textsc{Select}_Q$, which are defined as  
\begin{align}
    \textsc{Select}_P = \sum_i \op{i} \otimes P_i \nonumber,\quad
    \textsc{Select}_Q = \sum_i \op{i} \otimes Q_i \nonumber
\end{align}
where $\ket{i}$'s are the (computational) basis states of the filtering qubits. 
Using $V\ket{j} = \sum_i v_{ij}\ket{i}$, one can easily check that, the output of the circuit (before the measurement) is 
\begin{align}
    \ket{\tt out} &= \left(\sum_i\sum_j v_{i0}v_{ij}^* \ket{j} \otimes Q_iUP_i\right)\ket{\psi} \nonumber\\
    &= \ket{0} \otimes \left(\sum_i p_i\, Q_iUP_i\right)\ket{\psi} + {\rm other\; terms} \nonumber.
\end{align}
Therefore, if the measurement outcome on the ancilla is~0, we can achieve the transformation of  \autoref{eq:filter-superchannel}. In general, for any Quantum Channel $\mathcal{E}$ with Kraus Elements $E_k$, this filter (probabilistically) transforms $\mathcal{E}$ into: 
\begin{equation}
    \sch{F}[\mathcal{E}]: E_k \xrightarrow{\; \ket{0} \;} \sum_i p_i \, Q_i E_k P_i.
\end{equation}
In the special case where all the $P_i$'s and $Q_i$'s are Paulis, we call $\sch{F}$ a \emph{Pauli filter}. If $\textsc{Select}_P = \textsc{Select}_Q$, we call this a \emph{symmetric} filter. 
Although this implementation of a superchannel is probabilistic, the success probability could be  improved by using amplitude amplification methods \cite{grover1996fast, brassard2002quantum}.

If we post-select on the outcome $\ket{j}$ on the ancilla, the complete super-Kraus representation of the general quantum filter $\sch{F}$ in \autoref{fig:general-filter} is given by 

\begin{equation}
    \sch{F}[\mathcal{E}]: \bs{F_j}[E_k] \xrightarrow{\; \ket{j} \;} \sum_i \alpha_{ij} \, Q_i E_k P_i
\end{equation}
where the coefficients are $\alpha_{ij} = v_{i0}v_{ij}^*$. Hence, this scheme gives us a way to implement a number of super-channels in one go. 

\subsection{Application: Clifford circuit purification using Pauli Filters}\label{sec:clifford-purification}

{In this section, we focus on a specific application of the (general) quantum filter constructed in the previous section in purifying Clifford Circuits. The filtering scheme we provide uses only Hadamards and controlled-Pauli operations, thus enabling us to purify (noisy) Clifford circuits using purely Clifford gates. }

Consider the following scenario: we are running a circuit compiled in the Clifford+T gate-set. The circuit is arranged in layers of Clifford gates and T gates in \autoref{fig:clifford+t_circuit}. We consider a local depolarising noise $\mathcal{D}_p$ of strength $p$ on each qubit: 
\begin{equation}\label{eq:depolarising-main}
    \mathcal{D}_p(\rho) = (1-p)\cdot \rho + p\cdot \frac{1}{3}(X\rho X + Y\rho Y + Z\rho Z).
\end{equation}
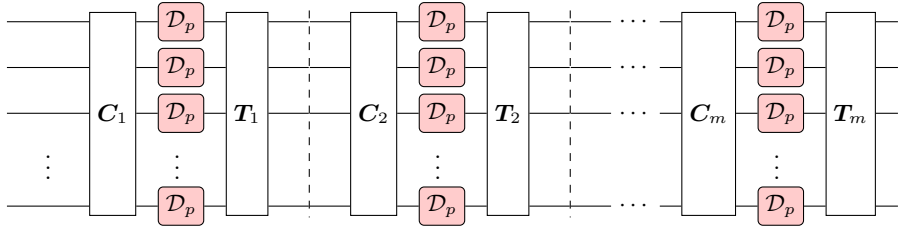
\begin{figure*}
    \begin{center}
\yquantdefinebox{dots}[inner sep=0pt]{$\vdots$}
    \begin{tikzpicture}[scale=1]
        \begin{yquant}
        \yquantset{operator/separation=3mm}
            qubit {} x[3]; 
            nobit k; 
            qubit {} y; 
            dots k; 
            box {$\bs{C}_1$} (x, k, y); 
            [shape=yquant-rectangle, rounded corners=.2em, fill=red!20]box {$\mathcal{D}_p$} x, y;
            dots k; 
            box {$\bs{T}_1$} (x, k, y); 
            barrier (x, k, y); 
            box {$\bs{C}_2$} (x, k, y);
            [shape=yquant-rectangle, rounded corners=.2em, fill=red!20]box {$\mathcal{D}_p$} x, y;
            dots k;
            box {$\bs{T}_2$} (x, k, y); 
            barrier (x, k, y); 
            [draw=white]box {$\cdots$} x, y; 
            box {$\bs{C}_m$} (x, k, y);
            [shape=yquant-rectangle, rounded corners=.2em, fill=red!20]box {$\mathcal{D}_p$} x, y;
            dots k; 
            box {$\bs{T}_m$} (x, k, y); 
        \end{yquant}
    \end{tikzpicture}
\end{center}
    \caption{Schematic of a Circuit Compiled in the Clifford+$T$ gate-set. $\bs{C}_i$ and $\bs{T}_i$ are circuit layers composed purely of Clifford and T gates respectively. $\ch{D}_p$ is a local depolarising channel with error-rate $p$.}
    \label{fig:clifford+t_circuit}
\end{figure*}

The circuit in \autoref{fig:clifford+t_circuit} has $m$ layers. Each layer $L_i$ is composed of a Clifford
unitary $\bs{C}_i$ and, a unitary $\bs{T}_i$ composed only of T Gates.
Each Clifford layer is followed by a layer of local depolarising noise. We assume that the number of T Gates in each $\bs{T}_i$ is small (compared to the number of gates in the $\bs{C}_i$'s), and thus the noise in them is negligible. 

If the number of qubits in the circuit is $n$, the combined channel after each Clifford layer is : 
\begin{align}
    \mathcal{D}_p^{\otimes n}\circ \mathcal{C}_i(\rho) =& (1-p)^n\cdot \bs{C}_i\rho \bs{C}_i^{\dagger} 
    + \sum_{w=1}^n (p/3)^w(1-p)^{n-w} \times
    \nonumber\\ 
    &\sum_k N_k[w]\bs{C}_i\rho \bs{C}_i^{\dagger}N_k[w]^{\dagger}
\end{align}
where $N_k[w]$ are Pauli operators with weight $w$~\footnote{A Pauli of weight $w$ acts non-trivially on $w$ qubits. For example, the Pauli $I_1X_2Y_3I_4Z_5$ has weight 3 and acts non-trivially on the qubits 2, 3, and 5.}, and
$
    \mathcal{C}_i(\rho) = \bs{C}_i\rho\bs{C}_i^{\dagger}
 $
is the (pure) Clifford channel. So, the components of the combined channel consists of the noiseless Clifford term $\bs{C}_i$ and noise-terms $N_k[w]\bs{C}_i$. In the following, we focus on a single such layer of Clifford circuits
under Pauli noise.

Consider one such component of the channel plugged into a Pauli filter
\begin{center}
\begin{tikzpicture}
\begin{yquant}
qubit {$\ket{0}$} q;
qubit {} r;
h q;
box {$P$} r|q;
hspace {0.1cm} r;
box {$\bs{C}$} r; 
[shape=yquant-rectangle, rounded corners=.2em, fill=red!20]box {$N$} r; 
hspace {0.1cm} r;
box {$Q$} r|q ; 
h q;
measure q;
discard q; 
\end{yquant}
\end{tikzpicture}
\end{center}
If $P$ is a Pauli operator, this circuit is equivalent to 

\begin{center}
\begin{tikzpicture}
\begin{yquant}
qubit {$\ket{0}$} q;
qubit {} r;
h q;
hspace {0.7cm} r;
box {$\bs{C}$} r;
hspace {0.1cm} r;
box {$\vec{P}$} r|q;
[shape=yquant-rectangle, rounded corners=.2em, fill=red!20]box {$N$} r; 
box {$Q$} r|q ; 
h q;
measure q;
discard q; 
\end{yquant}
\end{tikzpicture}
\end{center}
where we have defined 
$
    \overrightarrow{P} := \bs{C}P\bs{C}^{\dagger}
 $, as the forward propagation (in time) of the Pauli $P$ through the the Clifford $\bs{C}$. Since the Clifford group is the normalizer of the Pauli group, $\overrightarrow{P}$ is also a Pauli operator. 

Now, if we select $P=\bs{C}^{\dagger}Q\bs{C}=\overleftarrow{Q}$ in the original filter circuit (which is the backward propagation of $Q$ through $\bs{C}$), then $\overrightarrow{P} = Q$ and the resulting circuit becomes equivalent to 
\begin{center}
\begin{tikzpicture}
\begin{yquant}
qubit {$\ket{0}$} q;
qubit {} r;
h q;
hspace {0.7cm} r;
box {$\bs{C}$} r;
hspace {0.1cm} r;
box {$Q$} r|q;
[shape=yquant-rectangle, rounded corners=.2em, fill=red!20]box {$N$} r; 
box {$Q$} r|q ; 
h q;
measure q;
discard q; 
\end{yquant}
\end{tikzpicture}
\end{center}
As a result, the noise-channel $N$ sees a symmetric Pauli filter $\sch{F}_Q$. In \autoref{sec:quantum-channel-correction}, we have already shown that an arbitrary noise channel can be {deterministically} purified with $2n$ ancillas by using a symmetric Pauli filter of the following configuration 
\begin{align}
    V &= H^{\otimes 2n} \nonumber \\
    \textsc{Select}_Q &= \sum_i \op{i}\otimes Q_i \label{eqn:select-q}
\end{align}
where the sum on $i$ runs over all the $n$-qubit Pauli strings. Thus, the noisy Clifford $\bs{C}$ will also be completely purified if we choose, 
\begin{equation}
    \textsc{Select}_P = \textsc{Select}_{\overleftarrow{Q}} = \sum_i \dyad{i} \otimes \bs{C}^{\dagger}Q_i\bs{C}\label{eqn:select-p}
\end{equation}

\change{The capabilities of a filter can be improved by adding more clean ancillas. To be specific, if we are given access to $k$ ancillary qubits, we can completely clean $k/2$ system qubits using this filter configuration. Hence, the output fidelity scales with the number of ancilla as 
\begin{align}
\label{eq:fidelity_decrease}
    F_o = (1-p)^{n-k/2} &= F_i\cdot\left(1/\sqrt{1-p}\right)^{k}\nonumber\\ &\sim F_i\exp((k/2n-1)\epsilon)
\end{align}
which is an exponential increase with respect to $k$, and $\epsilon = 1-p$ is the input infidelity.} 

\change{The conventional QEM works cannot reduce the error in expectation values by gradually adding more qubits because the elementary block is a query to the whole channel.
A few works may achieve error reduction with an increasing number of qubits but with different resource costs than ours. For example, in the multi-round purification scheme of \cite{childs2023streaming} the fidelity increases exponentially with each round. However, the number of ancilla qubits and copies of the input state required also increase exponentially, thus making the fidelity increase linear in the number of ancillas. In \cite{lee2023error}, the decrease in infidelity is exponential in the number of ancillas. However, this suffers from the requirement of exponentially many calls to the original circuit and the requirement of non-Clifford gates. We achieve the above scaling using only a single query and Clifford gates for noisy Clifford circuits. In addition, the other schemes work in a post-selected manner, and therefore, the expected scaling is likely to be much worse. In contrast, we achieve our scaling deterministically with the maximum number of ancillas saturating $2n$.}

\change{Similar to the contemporary works, our result holds under the assumption of noiseless ancillas. An analytical treatment of the effect of ancilla-noise can be found below with more details in Appendix \ref{appendix:error-analysis}, and a corresponding numerical exploration in \autoref{sec:numerics}.  Below, we present the result for fidelity when considering the ancilla noise.
For any noisy 2-qubit gate, we model the noise by assuming that any application of the gate is followed by a local depolarising noise channel defined in \autoref{eq:depolarising-main} on the control and target qubits.
If $\ch{U}_{i;j}$ is an ideal 2-qubit gate (with control on the $i$-th qubit and target on the $j$-th qubit), we model the noisy version of this gate $\Tilde{\ch{U}}_{i;j}$ as 
$$
    \Tilde{\ch{U}}_{i;j} = \left(\ch{D}_{p_c}^{(i)} \otimes \ch{D}_{p_t}^{(j)}\right) \circ \ch{U}_{i;j}
$$
where $p_c$ and $p_t$ are error rates on the control and target qubits, respectively. For this noise model on the 2-qubit operations, if we want to purify a local depolarising channel on a single qubit using the channel correction filter in \autoref{fig:channel-correction}, the fidelity of the corrected channel is lower bounded by 
\begin{equation}
\label{eq:fidelity_ancilla}
    \Tilde{F} \geq (1-p_c)(1-p_t)(1-2p_c/3)^2(1-2p_t/3)^2 (1-2p_c/3).
\end{equation}
A derivation of this bound is given in Appendix \ref{appendix:error-analysis}. This result indicates that the output fidelity is more sensitive to noise on the ancilla. The above lower bound gives a pseudo-threshold on the error rate of the noise channel, above which using the filter is beneficial. }

\subsection{Implementation of \texorpdfstring{{\sc Select}$_P$}{Select-P} and \texorpdfstring{{\sc Select}$_Q$}{Select-Q}}

A problem of practical matters needs to be addressed regarding implementing the unitaries $\textsc{Select}_P$ and $\textsc{Select}_{Q^{}}$ in \autoref{eqn:select-q}, \autoref{eqn:select-p}. If the $2n$ filtering qubits are labelled as $f_{1, 1}, f_{1, 2}, f_{2, 1}, f_{2, 2}, \ldots f_{n, 1}, f_{n, 2}$, and the system qubits as $s_1, s_2, \ldots s_n$, then, one choice of $\textsc{Select}_Q$ is 
    \begin{equation}
        \textsc{Select}_Q = \bigotimes_{i=1}^n \textsf{C}_{f_{i, 1}}Z_{s_i} \cdot \textsf{C}_{f_{i, 2}}X_{s_i} 
    \end{equation}
    where $\textsf{C}_{f_{i, 1}}Z_{s_i}$ is a controlled-$Z$ gate with control-qubit $f_{i, 1}$ and target-qubit $s_i$, etc. This requires $n$ $CZ$ and $n$ $CX$ gates.

    Then, to correct a Clifford $\boldsymbol{C}$, the corresponding $\textsc{Select}_{P}$ will be 
    \begin{align*}
        \textsc{Select}_{P} &= \bs{C}^{\dagger}(\textsc{Select}_Q)\bs{C}\\
        &=\bigotimes_{i=1}^n \bs{C}^{\dagger}(\textsf{C}_{f_{i, 1}}Z_{s_i})\bs{C} \cdot \bs{C}^{\dagger}(\textsf{C}_{f_{i, 2}}X_{s_i})\bs{C} \\
        &= \bigotimes_{i=1}^n \textsf{C}_{f_{i, 1}}\overleftarrow{Z}_{s_i} \cdot \textsf{C}_{f_{i, 2}}\overleftarrow{X}_{s_i}
    \end{align*}
    where 
    $ 
        \overleftarrow{Z}_{s_i} = \bs{C}^{\dagger}Z_{s_i}\bs{C} $,  $\overleftarrow{X}_{s_i} = \bs{C}^{\dagger}X_{s_i}\bs{C} $
    are the backward propagation of the single-qubit Pauli-$X$ and Pauli-$Z$'s through the Clifford $\bs{C}$. 
    
    In the worst case, each $\overleftarrow{Z}_{s_i}, \overleftarrow{X}_{s_i}$ can be full-weight Paulis even though $Z_{s_i}, X_{s_i}$ are single-weight Paulis. So, each $\textsf{C}_{f_{i, 1}}\overleftarrow{Z}_{s_i}$ may require $n$ controlled-$X$ gates to implement. As a result, $2n^2$ $CX$-gates are needed to implement the entire $\textsc{Select}_{P}$ unitary.
    
  Note that both the unitaries $\textsc{Select}_{P}$ and $\textsc{Select}_{Q}$ are composed entirely of Clifford gates. Combined with the fact that the PREPARE unitaries of the ancilla for a Pauli filter are $H^{\otimes 2n}$, we can see that the entire Pauli filter is a Clifford circuit. Therefore,  {we can correct Clifford circuits using purely Clifford circuits}.


\subsection{Partial purification of non-Clifford gates}
\label{sec:partial-non-clifford}
In the above section, we have outlined how we can purify Clifford circuits using Clifford gates, particularly controlled-Paulis. This scheme can be adapted to the purification of non-Clifford gates such as T Gates, {and more generally, $R_z(\theta) = e^{-i \theta Z}$ gates which have applications in quantum dynamics simulation by Trotterisation or eigenstate preparation \cite{georgescu2014quantum}.} In this case,   non-Clifford gates are needed to completely purify circuits composed of T Gates. In the regime of fault tolerant quantum computing, this is an expensive resource where Clifford gates can be considered almost free. Here, we outline a method to {partially} purify T Gates using less expensive Clifford gates. 

Consider the following circuit 

\begin{center}
\begin{tikzpicture}
\begin{yquant}
qubit {$\ket{0}$} q;
qubit {} r;
h q;
box {$Z$} r|q; 
box {T} r; 
[shape=yquant-rectangle, rounded corners=.3em, fill=red!20]box {$\ch{D}$} r; 

box {$Z$} r|q ;  
h q;
measure q;
box {$X$} r|q;
discard q;

hspace {0.5cm} r;
\end{yquant}
\end{tikzpicture}
\end{center}
where we model a noisy T-gate as an ideal T-gate followed by a local depolarising error model $\mathcal{D}$
\begin{equation}\label{eqn:general_dep_error}
    \mathcal{D} = (1-p) \ch{I} + p_X \ch{X} + p_Y\ch{Y} +  p_Z\ch{Z}
\end{equation}
with $p_X + p_Y + p_Z = p$, and $\ch{X}(\rho) = X\rho X$.
{This error model is motivated by probabilistic	compilation techniques~\cite{Kliuchnikov2023shorterquantum,koczor2024sparse,PhysRevLett.132.130602}.}
Since T commutes with $Z$, this circuit is equivalent to
\begin{center}
\begin{tikzpicture}
\begin{yquantgroup}
    \registers{
    qubit {} q; 
    qubit {} r; 
    }
    \circuit{
    discard q; 
    box {T} r;
    hspace {0mm} q, r;
    
    init {$\ket{0}$} q;
    h q; 
    box {$Z$} r|q; 
    [shape=yquant-rectangle, rounded corners=.3em, fill=red!20]box {$\ch{D}$} r; 
    
    box {$Z$} r|q ;  
    h q;
    measure q;
    box {$X$} r|q;
    discard q;
    }
    \equals[$\equiv$]
    \circuit{
    discard q; 
    box {T} r; 
    [shape=yquant-rectangle, rounded corners=.3em, fill=blue!20]box {$\Tilde{\ch{D}}$} r; 
    }
\end{yquantgroup}



\end{tikzpicture}
\end{center} 

{Hence, the noise-channel $\ch{D}$ sees the symmetric Pauli filter $\sch{F}_Z$. As shown before, $\sch{F}_Z$ converts $\ch{D}$ into the channel $\Tilde{\ch{D}}$, where 
\begin{equation}\label{eqn:biased_dep_channel}
    \Tilde{\ch{D}} = \sch{F}_Z[\ch{D}] = (1-p + p_X)\ch{I} + (p_Y + p_Z)\ch{Z}
\end{equation}
Thus, the fidelity of the channel increases by $p_X$. {This approach is particularly relevant for} architectures where the noise is biased \cite{aliferis2008fault, webster2015reducing}, i.e, where one of the components among $p_X, p_Y, p_Z$ is much greater than the other two. So, if the noise on the system qubits is biased in such a way that $p_X \gg p_Z, p_Y$, then, 
\begin{equation*}
    \Tilde{\ch{D}} \sim \ch{I}
\end{equation*}
and we have an (almost) perfect T-gate.
Similarly, if the dominant error is the $Y$ error, i.e, $p_Y \gg p_X, p_Z$, then we could replace the feedback operation of $CX$ after the measurement with $CY$, as below,}

\begin{center}
\begin{tikzpicture}
\begin{yquant}
qubit {$\ket{0}$} q;
qubit {} r;
h q;

box {$Z$} r|q; 
box {T} r; 
[shape=yquant-rectangle, rounded corners=.3em, fill=red!20]box {$\ch{D}$} r; 

box {$Z$} r|q ;  
h q;
measure q;
box {$Y$} r|q;
discard q;

hspace {0.5cm} r;

\end{yquant}
\end{tikzpicture}
\end{center} 

This will result in the channel, 
\begin{equation}
    \Tilde{\ch{D}} = (1-p + p_Y)\ch{I} + (p_X + p_Z)\ch{Z} \sim \ch{I}.
\end{equation}


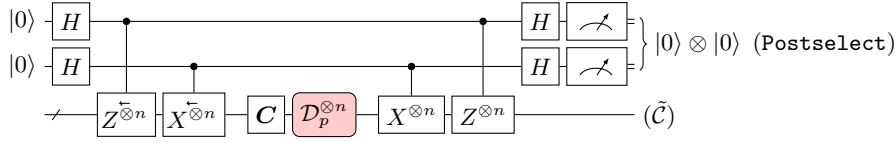
\begin{figure*}[t]
    \begin{center}
    \begin{tikzpicture}
    \begin{yquant}
    qubit {$\ket{0}$} q1;
    qubit {$\ket{0}$} q2;
    qubit {} r;
 
    h q1, q2;

    slash {n} r;
    box {$\backvec{Z^{\otimes n}}$} r|q1; 
    box {$\backvec{X^{\otimes n}}$} r|q2; 
    
    hspace {0.2cm} r;
    box {$\bs{C}$} r; 
    [shape=yquant-rectangle, rounded corners=.3em, fill=red!20]box {$\mathcal{D}_p^{\otimes n}$} r; 
    hspace {0.2cm} r;

    box {$X^{\otimes n}$} r|q2; 
    box {$Z^{\otimes n}$} r|q1;
    
    h q1, q2;
    measure q1;
    measure q2;
    
    hspace {0.5cm} r;
    output {$\ket{0}\otimes \ket{0}\; ({\tt Postselect})$} (q1, q2);
    output {$(\Tilde{\ch{C}})$} r; 
    \end{yquant}
    \end{tikzpicture}
    \end{center}
    \caption{An ancilla-efficient Pauli filter for purifying a Clifford Circuit $\bs{C}$ (affected by an $n$-qubit local depolarising channel). $\backvec{P} =  \bs{C}^{\dagger}P\bs{C}$ is the backward propagation (in time) of the Pauli $P$ through the Clifford $\bs{C}$.}
    \label{fig:ancilla-efficient-filter}
\end{figure*}

A depolarising channel $\ch{D}_p$ with $p_X = p_Y = p_Z = \frac{1}{3}p$ is converted into the {infinitely biased channel as}
\begin{equation}
    \Tilde{\ch{D}}(\rho) = (1-2p/3 )\cdot \rho + 2p/3\cdot Z\rho Z,
\end{equation}
{which is therefore amenable to techniques developed for biased noise models \cite{tuckett2018ultrahigh,bonilla2021xzzx,xu2023tailored}.}

{Beyond biasing the error channel, the fidelity is increased by $p/3$ as the incoherent $X$ and $Y$ components of the noise channel are removed}. {If we purify $n$ T Gates in parallel using this scheme (with $n$ ancillas), the residual noise-channel will be $\Tilde{\ch{D}}^{\otimes n}$ with a fidelity of $(1-2p/3)^n$,   which is a $1/3$ reduction in the decay rate compared to the raw channel. This is still exponentially small in $n$, a signature of the partial nature of the purification. 
This scheme can be adapted to partially purify the 3-qubit CCZ gate, which is shown in Appendix \ref{app:partial-ccz}. Subject to depolarising noise model on each qubit described above, this increases the fidelity of a CCZ gate from $(1-p)^3$ to $(1-2p/3)^3$. 
}

{The approach taken for non-Clifford gates in this subsection shares a similar spirit to quantum error correction enhanced quantum metrology - it is nearly impossible to protect against errors along the same axis as the signal (which may produce arbitrarily small rotation angles), but if the error and the signal occur in orthogonal directions, then error correction can mitigate the effects of noise in a similar vein in quantum metrology~\cite{shettell2021practical,zhou2018achieving}.

\change{To circumvent this, we could modify the T gate itself. In some physical devices, the dominant noise is the $Z$ error. 
Instead, we can realise the T gate by $H - R_x(\pi/4) - H$ gate with $ R_x(\pi/4) $ being a rotation gate along the $x$ direction.  If the Hadamard gate can be implemented fault-tolerantly, then by using similar techniques outlined above, we can completely remove the $Z$ errors.
Specifically, by sandwiching the $R_x(\pi/4)$ gate between two \textsc{cnot} gates, the $Z$ error can be detected. 
Another way to remove the $Z$ error is by concatenating with a standard repetition code. Note that now the error is just a dephasing error along the $Z$ direction, which can be corrected by a classical error correction code, like a repetition code.
}

}

\subsection{ Conjecture of efficient filtration} 

Using $2n$ additional qubits to {cancel errors for an} an $n$-qubit channel might seem expensive in some cases. If we want to {deterministically} remove all the noise components of a general $n$-qubit Pauli Channel, unfortunately, $2n$ filtering qubits are {necessary}. This is because, in order to {remove noise from} such a channel we need $4^n$ distinct syndromes to distinguish the $4^n$ Pauli Channel components (so that we can perform the correction for the corresponding component). To do this, we need at least $\log_2(4^n) = 2n$ qubits. 

However, for the local depolarising noise model considered above, all the noise components {do not} have equal strength. The strength of a (Pauli) noise-component with weight $w$, is 
$ 
    p^w(1-p)^{n-w}.
$
Assuming that the {depolarising strength} $p$ is small, components with very high weights will be less likely. Therefore, we may be content with removing components with weight up to some threshold $k$. The number of $n$-qubit Paulis with weight $\leq k$ is, 
$$
    \sum_{i=0}^k \binom{n}{i}3^i \sim O\left(3^{k}\binom{n}{k}\right).
$$
The number of ancillary qubits that are needed to correct channel components up to order $k$ {is thus lower bounded by }
 \begin{equation}
 \label{eq:ancilla_num}
    \#(\text{ancilla}) \sim \Omega\left(\log_2\left(\frac{3en}{k}\right)^k\right) \sim \Omega(k\log_2(n/k)),
\end{equation}
where the inequality $   \binom{n}{k} < \left(\frac{\text{e}n}{k}\right)^k
 $ is used. It remains to be seen whether this number of {ancillary} qubits is {sufficient} for stochastic Pauli noise correction.
 

\section{Ancilla-efficient Pauli filter}
\label{sec:good-enough}


In near-term applications, instead of correcting a channel, {probabilistically} filtering noise components using fewer resources might be desirable. In this section, we give a filter configuration that can remove weight-1 Paulis in any $n$-qubit Clifford Circuit using just 2 extra qubits. This is useful when weight-1 Paulis are the dominant noise component (which is the case in the regime $n< 1/p$), and the ancilla overhead is very small. 

Consider the following filter $\sch{F}_{AE}$ specified by:
\begin{equation}\label{eqn:ae-definition1}
    \begin{aligned}
         V &= H^{\otimes 2}\\
    \textsc{select}_Q &= \op{0}\otimes I^{\otimes n} + \op{1}\otimes Z^{\otimes n} \\
    &+ \op{2}\otimes X^{\otimes n} + \op{3}\otimes Y^{\otimes n}\\
    \textsc{select}_P &= \textsc{select}_{Q^{\star}} = \bs{C}\left( \textsc{Select}_{Q}\right) \bs{C}^{\dagger}\\
    \end{aligned}
\end{equation}
This filter maps a noise component $N$ to 
\begin{equation}\label{eqn:ae-pauli-map}
    N \xrightarrow{\sch{F}_{AE}} I^{\otimes n}NI^{\otimes n} + Z^{\otimes n}NZ^{\otimes n} + X^{\otimes n}NX^{\otimes n} + Y^{\otimes n}NY^{\otimes n}
\end{equation}
If the noise component $N$ is a weight-1 Pauli $S_i$, $S\in \{X, Y, Z\}$, which acts non-trivially only on the $i$-th qubit, we have, 
\begin{equation}
    S_i \xrightarrow{\sch{F}_{AE}} I_iS_iI_i + Z_iS_iZ_i + X_iS_iX_i + Y_iS_iY_i \quad \forall i
\end{equation}
This sum is always 0, since each weight-1 $S_i$ commutes with 2 of these terms and anti-commutes with the other two. Therefore, all the weight-1 Pauli's are filtered out. 
This results in an improvement in the fidelity of the channel.

\begin{figure*}[ht]
    \centering
    \includegraphics[width=0.8\textwidth]{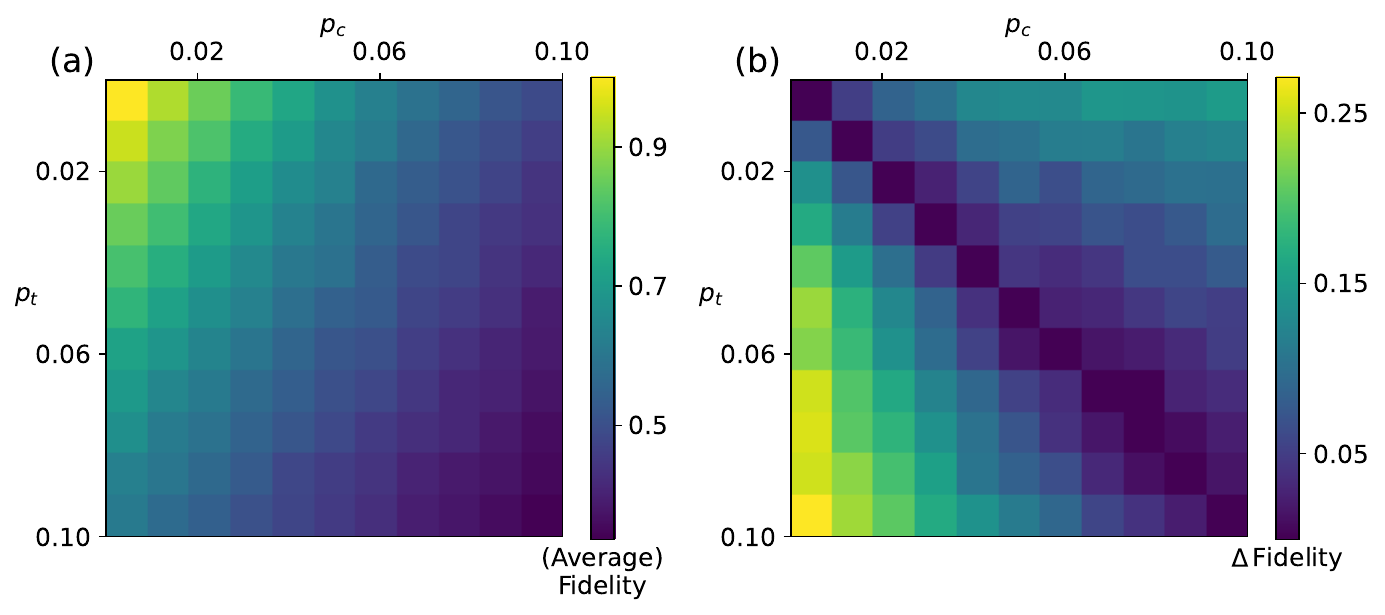}
    \caption{Numerical results for noisy implementation of the channel correction filter (\autoref{fig:channel-correction}) to purify local depolarising noise on 4 qubits. (a) Average Fidelity of the corrected channel for various combinations of $(p_c, p_t)$. (b) Fidelity gain, defined as $\Delta F = F(p_c, p_t) - F(p^*, p^*)$ with $ p^*= \max(p_c, p_t)$, showing the relative gain in fidelity for the error-rates on control qubit vs. target qubit, compared to the case where the error-rates on both the system and the ancilla are the same. As can be seen, the filter performance is more sensitive to noise on the ancilla (i.e, control qubits).}
    \label{fig:graphs-channel-correction}
\end{figure*}
 
\begin{figure*}[ht]
    \centering
    \includegraphics[width=\textwidth]{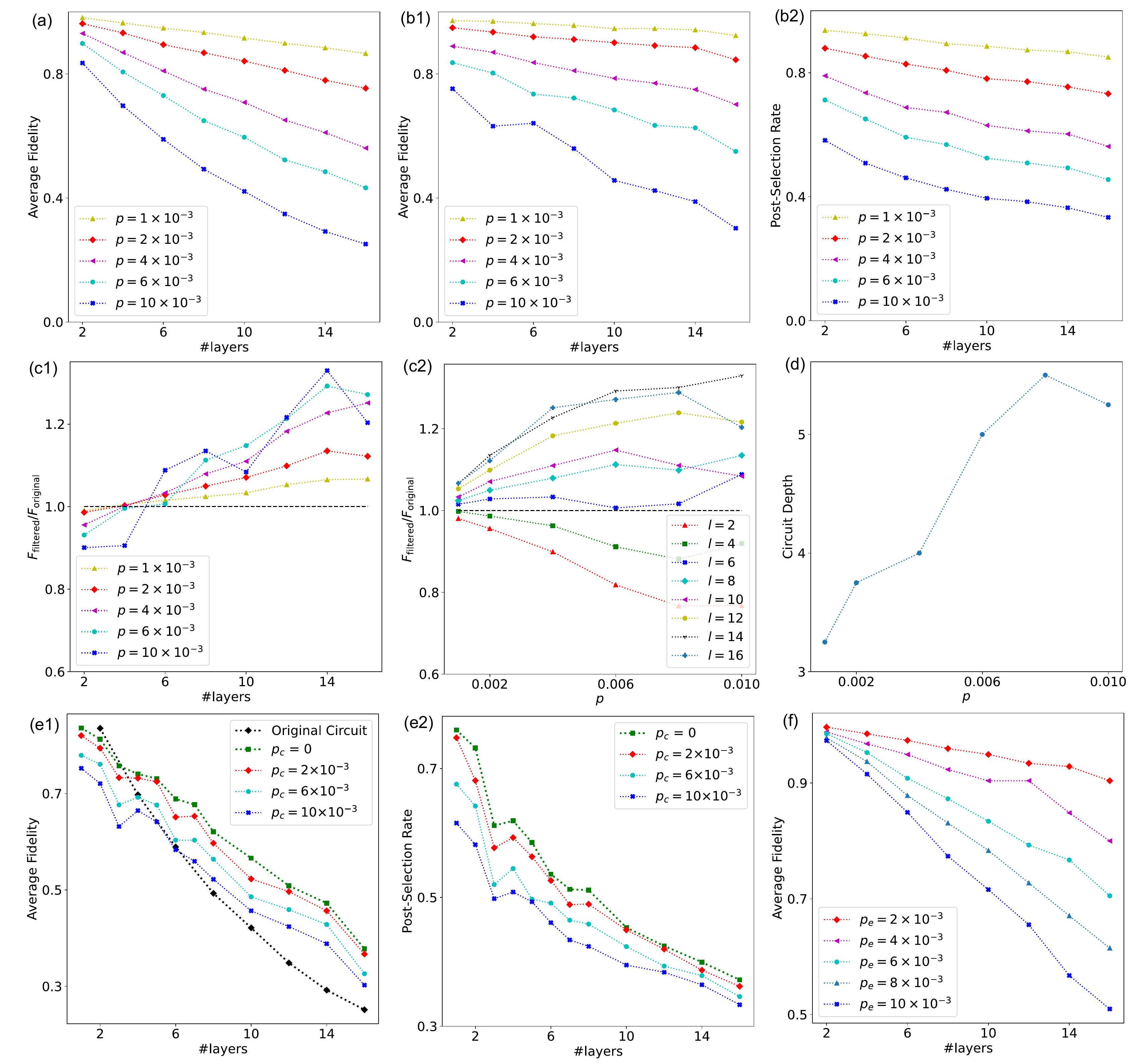}
    \caption{Performance of the ancilla-efficient Pauli filter (\autoref{fig:ancilla-efficient-filter}) for filtering layered brickwork circuits of 12 qubits {composed of \textsc{cnot} gates only}. For the plots (a)-(d), we use a common noise model for the filtering qubits and system qubits, and, uniform error-rates on the control and target qubits of 2-qubit gates, i.e, we set $p_c({\rm filter}) = p_t({\rm filter}) = p_c({\rm circuit}) = p_t({\rm circuit}) = p$. (a) Average Fidelity vs Circuit Depth for the original (unfiltered) circuit. (b1) Average Fidelity of the Filtered Circuit, (b2) Post-selection Rate, or, Success Probability of the ancilla-efficient filtering protocol. (c1), (c2) show the dependence of the Gain in Fidelity ($F_{\rm filtered}/F_{\rm original}$) vs the number of layers and error-rate respectively. (d) Crossover Depth (depth at which fidelity of the filtered circuit surpasses the fidelity of the original circuit) vs error-rate. Plots (e1),(e2) show the performance of the filter, if the filtering qubits are {cleaner} than the system qubits. They show dependence of the fidelity on $p_c({\rm filter})$, where we set $p_t({\rm filter}) = p_c({\rm circuit}) = p_t({\rm circuit}) = 0.01$. Finally plot (f) shows the performance of a noiseless filter, i.e, $p_c({\rm filter}) = p_t({\rm filter}) = 0$ and $p_c({\rm circuit}) = p_t({\rm circuit}) = p_e$.}
    \label{fig:numerics_ae_filter}
\end{figure*}

Now, we consider the fidelity improvement when considering local depolarising noise.
{The (average) fidelity of a noisy channel $\tilde{\ch{V}}$ compared to the ideal channel $\ch{V}$ is defined as 
\begin{equation}
    F_{\rm avg}(\tilde{\ch{V}}; \ch{V}) = \int \dd{\psi} \Tr[\tilde{\ch{V}}(\psi) \ch{V}(\psi)] 
\end{equation}
where the integral is over all pure states $\psi = \dyad{\psi}$ in the state-space. For a local depolarising error channel of the form of \autoref{eq:depolarising-main}, for which the ideal channel is the identity, the average fidelity reduces to  
$ 
    F_{\rm avg}(\ch{D}_p; \ch{I}) = 1 - \frac{2}{3}p
$, where $p$ is the probability of an error occurring \cite{horodecki1999general}. 
The average {infidelity} then becomes 
$ 
    \epsilon_{\rm avg} = 1 - F_{\rm avg} = \frac{2}{3}p = \Theta(p)$. 
In the following analysis, we ignore the unimportant scaling factor of $2/3$, and define the infidelity of a stochastic Pauli error channel as
\begin{equation}
\label{eq:infidelity}
    \epsilon = p_{\rm error} := 1 - p_{\ch{I}}
\end{equation}
where $p_{\ch{I}}$ is the probability of the identity component of the channel. Using this metric, the removal of weight-1 Paulis by the ancilla-efficient filter reduces the infidelity of the channel according to the following theorem. 
}

\begin{theorem}[{Quadratic reduction in the
infidelity with ancilla-efficient Pauli filters}]
Consider the ancilla-efficient Pauli filter $\sch{F}_{AE}$ defined by the unitaries in \autoref{eqn:ae-definition1} 
which can purify an $n$-qubit Clifford Circuit $\bs{C}$ using 2 ancillary qubits (\autoref{fig:ancilla-efficient-filter}). If the noise model on the system qubits is such that, the Clifford $\bs{C}$ is followed by an $n$-qubit local depolarising channel $\ch{D}_p^{\otimes n}$ of strength $p$, where  $D_p$ is defined in \autoref{eq:depolarising-main} and an input infidelity $ 
    \epsilon_{\tt in} := 1 - (1-p)^n
$ defined in \autoref{eq:infidelity}. {Then, the ancilla-efficient Pauli filter succeeds with a probability $p_s$ that is upper bounded by 
\begin{equation}
    p_s \leq 1 - p^2/\epsilon_{\tt in}
\end{equation}
}
Conditioned on its success, the infidelity of the purified channel is bounded by 
\begin{equation}
    \epsilon_{\tt out} \leq 2\epsilon_{\tt in}^2 = O(\epsilon_{\tt in}^2)
\end{equation} 
\end{theorem}
Hence, there is at least a quadratic reduction in the infidelity. A proof of this result can be found in Appendix \ref{app:quadratic_reduction}.  
{In a recent paper by Lee \etal\cite{lee2023error}, they have shown a black-box error suppression method which can be modified to achieve a quadratic reduction in the infidelity using $O(\log(1/\epsilon))$ ancilla and $O(1/\epsilon)$ queries to the circuit. Our scheme can do this with only 2 ancillas and 1 query, although for the specific noise model that we consider and only for the restricted case of Clifford circuits. Nevertheless, this is an important class of circuits. If we restrict the ancilla count to 2 qubits, Lee \etal's scheme achieves an output infidelity of $\epsilon_{\tt out}=\epsilon_{\tt in}/4$. Thus, our ancilla-efficient filter outperforms their scheme in the regime  $2\epsilon_{\tt in}^2 < \epsilon_{\tt in}/4$, i.e, when $\epsilon_{\tt in} < 1/8$. In Appendix \ref{section:global-pauli-noise-model}, we show that this result also holds for an exponentially decaying global Pauli noise model.}

All of these bounds are conservative, and the ancilla-efficient filter performs even better since it can remove many of the higher-weight Paulis as well. 

\begin{theorem}[{Filtration of low-weight Pauli errors}] Suppose the ancilla-efficient Pauli filter $\sch{F}_{AE}$ defined in \autoref{eqn:ae-definition1}  is used to purify an $n$-qubit stochastic Pauli channel (with $4^n$ Pauli components). Let $\#(w)$ be the number of Pauli-components of weight $w$ that $\sch{F}_{AE}$ can remove (if successful). Then, $\#(w)$ is lower-bounded by 
    \begin{equation}
        \#(w) \geq \binom{n}{w}3^w/w^2
    \end{equation} 
 $\sch{F}_{AE}$ can remove at least $4^n/n^2$ of all the $4^n$ Pauli components of the channel. Moreover, the average weight of all the Paulis that can be removed is (asymptotically) $n/2$.
\end{theorem}


Proof of this theorem and a more detailed treatment of the ancilla-efficient filter $\sch{F}_{AE}$ can be found in Appendix \ref{app:ge-filter}. {We note that the ancilla-efficient filter configuration given in \autoref{eqn:ae-definition1}, gives significant error reduction only in the regime where most of the error is concentrated in the low-weight Paulis. 
In cases where there are correlated errors over multiple qubits (resulting in error-concentration in high-weight Paulis),  $\sch{F}_{AE}$ may be less effective. 
}

\section{Numerical Results}
\label{sec:numerics}

In this section, we present numerical results for noisy implementations of the quantum filters when considering the noise on ancillas.  

The   lower bound of the fidelity in \autoref{eq:fidelity_ancilla} gives a pseudo-threshold on the error rate of the noise channel, above which using the filter is beneficial.  \autoref{fig:graphs-channel-correction} supports this   characteristics.

{  \autoref{fig:graphs-channel-correction}(a) shows the average fidelity of the output state for purifying local depolarising noise on 4 qubits, for various combinations of $p_t$ and $p_c$. \autoref{fig:graphs-channel-correction}(b) shows the relative gain in fidelity $\Delta F = F(p_c, p_t) - F(p^*, p^*)$, where $p^* = \max(p_c, p_t)$, obtained by keeping either the ancillas or the system qubits cleaner. As numerical evidence suggests, for the same error budget, it is more advantageous to keep the controls cleaner than the targets. }

Next, we provide some performance results for the ancilla-efficient filter. The circuits that we consider for filtering are layered brickwork circuits composed of \textsc{cnot} gates, which are simply used for benchmarking.

We consider 3 scenarios (i) the filtering qubits are subject to the same noise model as the system qubits ($p_c=p_t >0$), (ii) the filtering qubits are cleaner than the system qubits ($0 < p_c < p_t$) and (iii) the filtering qubits and operations are noiseless ($p_c=p_t=0$). We compare the dependence of the fidelity of the filtered circuits on various parameters. 

The results for the regime (i) are shown in the plots \autoref{fig:numerics_ae_filter}(a)-(d), where we set $p_c({\rm filter}) = p_t({\rm filter}) = p_c({\rm circuit}) = p_t({\rm circuit}) = p$. \autoref{fig:numerics_ae_filter}(a) shows the average fidelity of the output state vs circuit depth for the original (unfiltered) circuit, (b1) the average fidelity of the output state for the filtered circuit, and (b2) the post-selection rate or the success probability of the ancilla-efficient filter. As can be seen, the filtered output is consistently better than the original output, although the post-selection rate decreases with the error rate as expected. \autoref{fig:numerics_ae_filter}(c1) and (c2) show the dependence of the Gain in Fidelity ($F_{\rm filtered}/F_{\rm original}$) vs the number of layers and error-rate respectively. The relative gain in fidelity is significantly higher for more noisy circuits and higher depths. \autoref{fig:numerics_ae_filter}(d) shows the critical depth or crossover depth at which the fidelity of the filtered circuit surpasses the fidelity of the original circuit. If the filter and system qubits are subject to the same noise-strength, the crossover happens at deeper circuits for higher error rates. 

\autoref{fig:numerics_ae_filter}(e1) and (e2) show the performance of the filter in regime (ii), if the filtering qubits are {cleaner} than the system qubits. They show dependence of the fidelity on $p_c({\rm filter})$, where we set $p_t({\rm filter}) = p_c({\rm circuit}) = p_t({\rm circuit}) = 0.01$. As before, for the same depth, this results in a higher output fidelity with a corresponding penalty in the post-selection rate. 

Finally, \autoref{fig:numerics_ae_filter}(f) shows the performance of the ancilla-efficient filter in regime (iii), where we set $p_c({\rm filter}) = p_t({\rm filter}) = 0$ and $p_c({\rm circuit}) = p_t({\rm circuit}) = p_e$. This version of the filter has the best performance of all the cases which aligns with our expectation. 

From all the results, it is evident that the performance of the filter is more sensitive to the error rate on the ancilla. In all these cases, we see that if the depth of the circuit is above a certain level, the ancilla-efficient filter provides a modest gain in fidelity. This holds true even if the filtering qubits are as noisy as the system qubits.

\section{Further remarks}


\new{In this section, we explore connections of our scheme to prior works in QEM, QEC and quantum error reduction as a whole, and outline aspects in which they significantly differ and/or outperform them. We note that the availability of low-error ancillas is essential to obtain the performance gains. As such, in the later part of this section, we discuss this assumption in more detail. We do a more in-depth analysis of the cost of our protocol taking these details into account and clarify its application in noise-biasing. Finally, we outline possible hardware configurations suitable to implement our protocol, discussing perspectives for both physical and logical qubits.}

\subsection{Interpretation and comparison in the context of QEM and QEC}

\change{{\bf Comparison with existing QEM schemes}. As stated above, the quantum filtering scheme overcomes the exponential sampling overhead of standard QEM protocols such as Zero Noise Extrapolation~\cite{temme_error_2017, li_efficient_2017} and Probabilistic Error Cancellation~\cite{endo_practical_2018}, which purely rely on classical post-processing. Thus, our method also performs better than hybrid QEC-QEM schemes such as~\cite{lostaglio2021error, piveteau2021error,suzuki2022quantum} which use ZNE or PEC as a subroutine, and thus also suffers from the exponential sampling overhead. Symmetry Verification~\cite{bonet2018low, mcardle2019error} is based on post-selection, but requires the measured state to have preexisting symmetries and does not provide a way to remove all the errors. Our quantum filtration scheme does not require the existence of these inherent symmetries, and can possibly remove all errors. Some recent works like Virtual State Purification~\cite{huggins_virtual_2021, koczor_exponential_2021}, Virtual Channel Purification~\cite{liu2024virtual}, Streaming Quantum State Purification~\cite{childs2023streaming}, and Black Box Error Suppression~\cite{lee2023error}  are variations of the \textsc{swap} Test and coherently control multiple copies of the noisy state or noisy channel to reduce errors. The virtual protocols can achieve an exponential suppression in error for each copy of the state/channel but can only recover expected values, not the noiseless state. The other protocols can recover the noiseless state but require exponentially many copies of the noisy state/channel.}

\change{To summarise,  our protocol can simultaneously achieve the following: (1) overcoming the exponential sampling cost of QEM, (2) using only a single query to the noisy circuit to be purified, and (3) only Clifford resources are required. 
We emphasize that none of the existing QEM protocols so far have the ability to achieve even two of the three characteristics aforementioned. For example, features (1) and (3) are incompatible within all the post-QEM algorithms. 
Our channel correction scheme gives an exponential increase in the output fidelity with each ancilla, with only a single query to the original circuit and a linear number of ancillas. More importantly, all of the other works require non-Clifford multiqubit controlled-\textsc{swap} gates and work in the post-selection mode. In contrast, our scheme requires only Clifford gates and can work deterministically.}

\begin{table*}[t]
\begin{tabular}{c|ccccc}
& \begin{tabular}[c]{@{}c@{}}\#Copies of \\ Noisy \\ Channel\end{tabular} & \#Ancilla        & \begin{tabular}[c]{@{}c@{}}\#Extra \\ Gates\end{tabular}                   & \begin{tabular}[c]{@{}c@{}}Output \\ Infidelity\\ (approx.)\end{tabular} & \begin{tabular}[c]{@{}c@{}}Works \\ for\end{tabular}                            \\ \hline\hline
\begin{tabular}[c]{@{}c@{}} Virtual/Streaming \\ State Purification~\cite{childs2023streaming, koczor_exponential_2021} \end{tabular}        & $k$                                                                     & $\sim nk$        & \begin{tabular}[c]{@{}c@{}}$\sim nk$ \\ controlled-\textsc{swap}s\end{tabular}      & $ 2^n\epsilon^k/2^{nk}$                                    & \begin{tabular}[c]{@{}c@{}}Non-Unitary Errors\\ (Arbitrary Channels)\end{tabular}        \\
\begin{tabular}[c]{@{}c@{}}Virtual Channel\\ Purification~\cite{liu2024virtual},\\ Superposed QEM~\cite{miguel2023superposed} \end{tabular} & $k$                                                                     & $\sim nk$        & \begin{tabular}[c]{@{}c@{}}$\sim 2nk$ \\ controlled-\textsc{swap}s\end{tabular}     & $4^n\epsilon^k/4^{nk}$                                    & (As above)                                                                             \\
Error Filtration~\cite{lee2023error}                                                                       & $k$                                                                     & $\sim n+ \log k$ & \begin{tabular}[c]{@{}c@{}}$\sim 2nk$\\ multi-controlled-\textsc{swap}s\end{tabular} & $\epsilon/k$                                          & (As above)                                                                             \\ \hline
\begin{tabular}[c]{@{}c@{}}Channel Correction\\ Pauli Filter\end{tabular}                 & $1$                                                                     & $k$              & \begin{tabular}[c]{@{}c@{}}$nk + k$\\ controlled Paulis\end{tabular}       & $1- \exp((\frac{k}{2n}-1)\epsilon)$                             & \begin{tabular}[c]{@{}c@{}}Arbitrary Errors \\ (Clifford Channels)\end{tabular} \\

\begin{tabular}[c]{@{}c@{}} Ancilla-Efficient\\ Pauli Filter\end{tabular}                                                         & 1                                                                       & 2                & \begin{tabular}[c]{@{}c@{}}$4n$\\ controlled Paulis\end{tabular}           & $\epsilon^2$                                          & (As above) \\ \hline\hline
\end{tabular}
\\
*For purifying an $n$-qubit state/channel.
\caption{Resource comparison of the Pauli Filter and contemporary Ancilla-assisted Quantum Error Mitigation schemes.}
\label{table:resource-comparison}
\end{table*}

\new{In \autoref{table:resource-comparison}, we give a comparative breakdown of the resource overheads of our scheme and the contemporary AQEM schemes. As can be seen, our protocol is an important addition to this class and provides notable improvements in qubit/gate overhead under our specific restrictions. 
}\\

{\bf Connection to QEC Codes.} 
The error detection circuit, which is part of the quantum filter, has certain similarities to circuits in flag fault-tolerant error correction (see \cite{chamberland2018flag,chao2020flag,chao2018quantum}) and coherent parity check codes~\cite{roffe2018protecting,van2023single}. 
Compared to coherent parity check codes,  we give a recipe for the checks that we need to perform in order to achieve complete purification. Similarities can also be drawn between the Clifford channel purification scheme in \autoref{sec:clifford-purification} and spacetime checks in the recently proposed spacetime codes of Clifford circuits \cite{delfosse2023spacetime}, although \autoref{theorem:commutation} gives a general condition that extends our scheme beyond Clifford channels. The construction of the circuit described in \autoref{fig:ancilla-efficient-filter} has similarities to the ‘Iceberg’ code used in \cite{self2024protecting}, Rains’ optimal distance-$2$ code \cite{rains1999quantum} for even $n$, and to entanglement-assisted QEC codes \cite{brun2014catalytic} for pure storage channels. It can also be viewed as a coherent example of the conversion between QEC codes and entanglement distillation schemes \cite{aschauer2005quantum} (for which the only requirement on the Pauli frame is that it matches between the two parties). \new{The procedure to derive the \textsc{Select} operations by conjugating Paulis with respect to Cliffords, is also used incoherently in Pauli Twirling. But, Pauli Twirling does not have any correction properties and thus cannot increase the fidelity of a channel, which our scheme can.} 

\new{Although one might draw cursory similarities of the circuit constructions in our protocol to those mentioned above, the underlying assumptions, interpretations and applicability of our scheme are completely different. A discussion on the connections and key differences of the quantum filter with these schemes is provided in Appendix~\ref{appendix:connections}, focusing on Pauli Twirling, the Iceberg Code, Entanglement-Assisted QEC Code and Flag-QEC. In Appendix~\ref{sec:internal-noise}, we give further numerical results on what happens to the performance of the Pauli Filter if we drop the assumption of noiseless ancilla, comparing the results with the Flag-QEC protocol based on the [[5, 1, 3]] code (which has the lowest qubit overhead) as a reference. We find that, even though our scheme is not fault-tolerant, due to the efficient qubit/gate overhead of our scheme, it outperforms Flag-QEC in tolerating internal noise in the moderate to high noise-regime. This is representative of the noise-regime in which our scheme can be useful compared to contemporary QEC methods.}

Compared to typical QEC codes, our method can detect or correct errors in a circuit without explicitly encoding the input state. These qubit-efficient QEC codes may be regarded as one limit case of our method (i.e. the ancilla-efficient Pauli filter), differing in their requirement for logical qubit encoding or fault-tolerant logical operations. In the other limit case, we can achieve channel correction without qubit encoding. 
The regions in between have not been explored by the existing schemes. However, it is essentially important to develop methods and run algorithms at this intermediate stage between NISQ and early fault tolerance. 

\new{
Moreover, a major effort in QEC design is to find efficient constructions of logical operations (between distributed qubits) with a moderate level of protection. The usual approach to this is to begin with a pre-defined list of codes and to search, largely by brute force, for a desired operation with a transversal set of gates \cite{sayginel2025fault}. Much recent work on magic-state distillation \cite{gidney2019efficient, vasmer2022morphing, akahoshi2024partially}, the construction of the Iceberg code \cite{self2024protecting} and even the Colour codes \cite{bombin2006topological} is based on this motivation,~e.g, \cite{self2024protecting} sacrifices fault-tolerance so that global logical operators can be implemented efficiently on the Iceberg Code using only the two-qubit M\o lmer-S\o rensen gate~\cite{molmer1999multiparticle}. Therefore, our proposal inverts this design procedure, which to some extent shifts the perspective in designing the error-correcting codes. If we view the quantum filter with the lens of QEC (with the \textsc{Select} operations as encoding/decoding maps), then our proposal starts with the operation that we want to purify (the ideal Clifford channel) and then constructs a filter around it, instead of searching for an efficient operator for a pre-defined code. This gives a conceptual advance that the logical operator can now be as simple as possible (even the same as the physical operator, as in our scheme, if interpreted in this context), while diverting the previous overhead to the encoding maps. }

Our quantum filtering scheme might have several attractive features compared to QEC {in terms of near-term applications}. 
{The scheme is input-state agnostic, and the filter construction depends only on the channel to be purified. As such, we do not need to encode the system qubits a priori, or convert the target channel into large multi-qubit logical gates.} We can plug in the original circuit of our computation into the filter as is, treating it almost like a black box.
Another advantage of our framework is that the capabilities of a filter can be tuned based on the amount of resources that are available, as evidenced  by \autoref{eq:fidelity_decrease} and \autoref{eq:ancilla_num}.
 
When the computation runs on the system qubits, the filtering qubits remain idle. They are acted upon only at the start and end of the computation. As a result, we can think of keeping the filtering qubits isolated while the computation is running, so that they are not affected by noise. In fact, it might also be possible to use two physically different types of qubits for these two registers, stable qubits with long coherence times for the ancilla, and qubits with fast gate execution for the main system.
In addition, since only specific types of gates act on the filtering qubits (Hadamard and controlled-Paulis in our Pauli Filtering scheme), and they don't take part in general-purpose computation, we can optimize these qubits only for these specific gates.\\

{\bf Different interpretations of our scheme.}
The construction of commutation filters can also be linked to other perspectives.
One can think of the idea as the use of a coherent variant of quantum phase estimation \cite{kitaev1995quantum}, with the syndrome measurements of quantum error detection schemes. The coherence allows us to measure only changes in syndrome, avoiding the complexity of encoding, at the cost of requiring high-quality ancillary qubits and controlled gates.
It also shows some similarity to the measurement of out-of-time-order-correlators (OTOCs) for scrambling \cite{mi2021information}, though the connection is not that straightforward.
It is possible to interpret our scheme as an application of indefinite causal order to error detection \cite{chiribella2021indefinite}, though with an additional Clifford pre-processing step.
 






{From a quantum communication viewpoint, we might think of scenarios where we have a small number of noiseless channels, and a large number of noisy channels. Then we can use the noiseless channels as filters to clean the noisy channels. This method can also be used as a quantum error mitigation technique where we can use a small number of clean qubits to perform gate-purification \cite{dur2003entanglement} on a larger number of noisy qubits. This paradigm of using extra qubits which are cleaner than the system qubits has also been explored under the moniker of error suppression \cite{lee2023error} and partial error correction \cite{koukoulekidis2023framework}. 
}

The quantum filters proposed in this work can map a quantum channel to another and thus provide a way to realise quantum channel transformations. Although we mainly discuss this in the context of channel purification, error reduction and correction, it can find applications in the context of implementing quantum superchannels \cite{chiribella2008transforming}. In particular, the general quantum filter construction introduced in \autoref{section:general-quantum-filter} can be used to realise a large class of quantum superchannels via LCU-inspired circuits \cite{childs2012hamiltonian}. Other types of quantum superchannels may be designed inspired by this idea, which could be an interesting direction for future work.


\subsection{Practical Considerations}

{\bf Assumption of Noiseless Ancilla.} In our protocol and the contemporary AQEM protocols, the idea is to make use of noiseless ancilla. Here, we outline details which need to be kept in mind while assuming the use of such qubits. 
\new{As briefly mentioned in \autoref{section:introduction}, scaling quantum devices beyond the NISQ era and into the early fault-tolerant era naturally gives rise to asymmetries in noise \cite{ding2021architecting, resch2021benchmarking}.  
Distributed and inhomogeneous quantum architectures might allow us to produce a few high-quality qubits but restrict us to construct a complete quantum infrastructure at this same level of quality \cite{ramette2024fault, main2025distributed} As such, the best use of such systems requires us to couple such asymmetric qubits and use them in synergy. Such coupling may be physically engineered and even done between error-corrected logical qubits and noisy physical qubits \cite{bultrini2023battle, koukoulekidis2023framework}. Designing protocols that take advantage of this setting will be an important task going forward.}

\new{However, an addendum in this regard is that even though one can assume the existence of clean (or, less-noisy) qubits in such settings, one must not assume that everything can be done with them. 
Indeed, if we had unlimited and unrestricted access to noiseless ancilla, we could just build our entire quantum computer out of them, and forget about the noisy qubits in the first place. 
We must assume that the asymmetry is accompanied by (or perhaps derived from) important restrictions: either the number of high-quality qubits is limited, or the type of gates that can be applied to them is limited, or both, or some other restriction.}

\new{In our protocol, we impose a strong restriction on the ancilla that only controlled-Pauli gates (with the ancillas as controls) are applied to them, as opposed to multi-qubit control gates. We access them only at the beginning and end of the computation, leaving them otherwise idle. For the ancilla-efficient filter, we further limit the number of ancillas. Though we restrict our attention to purifying Clifford channels, our requirement for ancilla operations is much simpler than those in other current AQEM protocols (where complicated multi-qubit controlled non-Clifford gates are required for implementation).}\\

\new{\bf Cost of {\sc Select} Operations and Implementation}.
The $\textsc{Select}$ operations needed to implement the Clifford Circuit purification scheme (outlined in \autoref{sec:clifford-purification}) require $2n^2 + 2n = O(n^2)$ 2-qubit gates (in the worst case). Current state-of-the-art methods for synthesizing $n$-qubit Clifford Operations also require $O(n^2)$ 2-qubit gates \cite{bravyi2021hadamard, schneider2023sat} (albeit with a higher pre-factor). As a result, although there is always some amount of error reduction due to the filter, if the filtering circuit and the main Clifford circuit are subject to the same strength of noise, the asymptotic rate of error-reduction might not be very high. This just highlights the point that the ancilla being cleaner than the main qubits is an essential assumption required to obtain significant gains in fidelity.

\new{In the interest of clarification, if we are targeting individual gate-elements in a circuit, the cost of the \textsc{Select} operations depends only on the number of qubits the gate acts upon non-trivially. For example, if a Clifford gate acts non-trivially on $m$ qubits in an $n$-qubit circuit, our scheme requires $O(m^2)$ \textsc{Select} operations, not $O(n^2)$. The filter can target individual gates without disturbing the other components in a circuit. More crucially, we do not need to apply the filter to individual gate elements either see figure below (left). Instead, we can apply the filter to a combined \emph{block} of Clifford gates (figure below (right)), saving the number of ancilla qubits and operations.}

\begin{figure}[!h]
\centering
\begin{yquantgroup}
\yquantset{operator/separation=0mm}
    \registers{
    qubit {} a;
    qubit {} b;
    }
    \circuit{
    slash a;
    init {$\ket{0}$} b;
    [style={fill=blue!20}, radius=2.5pt] box {} (a, b);
    box {$C_1$} a;
    [style={fill=blue!20}, radius=2.5pt] box {} (a, b);
    discard b;

    init {$\ket{0}$} b;
    [style={fill=blue!20}, radius=2.5pt] box {} (a, b);
    box {$C_1$} a;
    [style={fill=blue!20}, radius=2.5pt] box {} (a, b);
    discard b;

    }
    \equals[\,]
    \circuit{
    slash a;
    init {$\ket{0}$} b;
    [style={fill=blue!20}, radius=2.5pt] box {} (a, b);
    box {$C_1$} a;
    hspace {1mm} a;
    box {$C_2$} a;
    [style={fill=blue!20}, radius=2.5pt] box {} (a, b);
    discard b;
    }
\end{yquantgroup}
\end{figure}

\new{Moreover, the cost of $2n^2$ is the worst case cost {due to the saturation of Clifford circuits}. So long as the block of gates is entirely Clifford, the cost of $\textsc{Selet}_P$ depends on the sparsity of the Stabilizer Tableau of the combined Clifford block, and not necessarily on the depth of the block or the number of gates in it. As such, for practical Clifford circuits with a sparse tableau, the cost might be significantly less than $2n^2$ and thus manageable. Such Clifford gates with a sparse tableau arise in applications, such as QAOA \cite{munoz2024low} and bit-reversal circuits \cite{bapat2021quantum}, for which our scheme could have low cost in reducing errors.
}

\new{Thus, for a circuit compiled in the Clifford+T gate-set, we can target the Cliffords block-by-block, while partially purifying the T-gates in the middle (according to the scheme in \autoref{sec:partial-non-clifford}). It is also possible to go beyond partial purification of the T-gates. In magic-state injection protocols, we can implement a T-gate in a circuit by using a purified magic-state and an auxiliary Clifford circuit, the purification circuits being Cliffords as well. If we apply the Pauli filter protocol to these auxiliary Clifford circuits in magic-state injection, we can also completely  purify T-gates (albeit indirectly due to magic-state injection), though this comes at a 2-qubit gate cost linear increasing with the number of T-gates in the circuit.}

For the ancilla-efficient filter, significant error reductions are obtained in the regime where $n < 1/p$ ($p$ being the single-qubit depolarising error rate) and where most of the error strength is concentrated in low-weight Paulis. For long-range entangling gates which quickly spread errors over multiple qubits, the reduction in error is likely to be small. Nevertheless, the ancilla-efficient scheme still gives a low-cost method to reduce errors in deep Clifford circuits. 

If the filter operations are noisy, any noise on the system part of the $\textsc{Select}$ operations can be considered as part of the target channel. Nevertheless, such a treatment would not work for noise on the filtering qubits. This, as mentioned before, is a limitation of our method for practical applications. Although we have provided a basic treatment of the effect of these errors on the accuracy of the filter in Appendix \ref{appendix:error-analysis}, it is instructive to ponder how to avoid or potentially correct such errors. This might lead to a theory analogous to fault tolerance, which might be a direction for future work.\\

{\bf Partial Purification and Biased Noise}. In practice, physical qubits can be affected by a biased noise channel (i.e, one of the Pauli $Z$ or $X$ components of the channel is exponentially suppressed compared to the other). Examples of such biased-noise qubits are superconducting fluxonium qubits \cite{pop2014coherent}, quantum-dot spin qubits \cite{watson2018programmable}, etc. If the noise is biased, the filter operations in our scheme can be simplified. For example, if the noise is biased towards $Z$ errors (for which $X$ and $Y$ errors are exponentially suppressed), we can remove all the $Z$ errors from $n$ qubits deterministically using $n$ Pauli $X$-Filters. Note that this requires $n$ ancillas and thus has a factor of 2 reduction from the complete correction filter.

Recently, there has been progress in designing QEC codes for biased-noise qubits, which are shown to have high code capacity \cite{tuckett2019tailoring,tuckett2018ultrahigh,xu2023tailored}. The partial purification scheme outlined in \autoref{sec:partial-non-clifford} can be used to convert an unbiased noise channel to a channel biased towards either  $Z$ or $X$ errors. \new{Therefore, if a uniform noise channel unbiases a biased-noise QEC code (or any circuit constructed with biased-noise in mind), the partial purification filter gadget can restore the bias of the noise channel.} Our scheme may also find applications in the context of multi-core quantum computing based on Bell pair injections, in which case biased noise is common as found in experiments \cite{jnane2022multicore}.

In applications where one needs to find the expectation values of observables (which is the case in many practical quantum algorithms, like VQE \cite{tilly2022variational} and quantum simulation \cite{georgescu2014quantum}), the final result is naturally robust to certain types of errors. More concretely, if an observable $H$ commutes with all the Kraus elements of a noise channel $\ch{E}$, then the expectation value $\expval{H}$ is invariant to the effect of $\ch{E}$ on a state $\rho$, i.e, 
$  \Tr[\ch{E}(\rho)H] = \Tr[\left(\sum_k E_k^{\dagger}E_k\right)\rho H] = \Tr[\rho H]
$.
In these cases, it is sufficient to remove the components of $\ch{E}$ which anti-commute with $H$. For example, if we want to find the expectation of an observable with only Pauli $Z$ operators, it will be robust to $Z$ errors and mainly be affected by $X$ and $Y$ errors.\\

\new{{\bf Hardware Perspectives on Anisotropic Noise.} Here, we briefly sketch qubit platforms in which highly anisotropic noise-rates may be physically engineered, first pointing to discussions of such anisotropic qubits already prevalent in existing literature and in hardware implementations.}

\new{In trapped-ion systems, we have examples of the use of multiple ion species to achieve local high-fidelity operations and long-lived memory alongside simpler photon entanglement generation \cite{inlek2017multispecies}. Given the complex energy level structures in trapped ion systems, the use of different energy levels to define different qubits has also been proposed \cite{feng2024realization}. Finally, the vibrational modes of a trapped ion array, though typically quite noisy, do constitute a quantum system in their own right and it has been proposed to use them in hybrid encoding schemes \cite{fluhmann2019encoding, cochrane1999macroscopically}.}

\new{In superconducting circuit quantum computing, structures corresponding to quantum subsystems with different properties are used to enable measurement and to mediate and control coupling between data qubits \cite{arute2019quantum}. This arrangement between linear and non-linear resonators to mediate coupling has also been reversed, both to achieve high connectivity \cite{naik2017random} and to use the linear resonator in an encoding scheme \cite{campagne2020quantum}.
     In donor and defect systems (for example, phosphorus donors in silicon \cite{pla2013high} or nitrogen-vacancy centers in diamond \cite{childress2005fault, dutt2007quantum}), electronic spin states are often used to achieve rapid gates and to generate entanglement, while longer-lived nuclear spins are used to store information. It has also been proposed to mediate interactions between donor qubits in silicon via quantum dots \cite{morello2020donor}.
     Similar ideas have been proposed even at the logical level, where one may couple noisy and error-corrected qubits \cite{bultrini2023battle, koukoulekidis2023framework}, or use different codes for storage and computation on the same system  \cite{breuckmann2017hyperbolic}.}

\new{In neutral atom arrays, in addition to photon entanglement generation \cite{young2022architecture}, different species of atom have been used to achieve non-destructive measurements for the detection of correlated noise \cite{singh2023mid}.}

\new{In the examples above we observe that explorations combining different qubit types are already underway.
These works discuss combining qubits of different qualities as part of the same physical platform, but we could also consider a hybrid setup combining qubits from different physical platforms, such as trapped ions and superconducting circuits.
Trapped ion qubits are usually less noisy at a gate level but have slow gate execution times, while superconducting platforms have faster gates but are noisier. Leaving aside clear and difficult challenges surrounding the connecting infrastructure, in such a hybrid system, one might imagine a scenario where superconducting qubits are used to run the main computation, with trapped ion qubits as dedicated ``noiseless'' ancillas, operating less frequently.}

\new{To briefly outline a well-known approach to such a hybrid system, we might imagine generating two entangled photons of distinct frequencies resonant with the respective transition frequencies of a superconducting qubit and trapped-ion qubit.
By mediating interactions through such photons, entanglement can be established and entangling gates implemented via teleportation-based protocols.
One implied difficulty of such a scheme is the potential need for frequency conversion if the ion and superconducting qubits have quite different energies.
This and similar difficulties motivate the restriction of inter-species interactions so that they are sparse with respect to the full set of gates in the circuit, and this in turn motivates the question of which sparse subset of errors we most need to detect.}\\

\new{{\bf Clean Ancillas in Early Fault Tolerance.} The case of perfectly clean physical ancilla qubits is only the most extreme scenario, useful to demonstrate our method. Our protocol could be applied equally well to logical qubits, and the target Clifford circuit can be generalised to any efficiently simulable subcircuit.
In the hybrid-qubit setting, if the data-ancilla interactions are sparse and the ancillas remain idle most of the time (as we have assumed), then we can have different encodings for the ancilla and data qubits.
A large code-distance would make the ancillas essentially ``clean'' compared to the data-qubits. With reduced overhead of fine-tuned ancilla codes for multi-logical-qubit Pauli measurements \cite{huang2023homomorphic}, such logical qubits might be generated more cheaply in the early fault-tolerant regime. Existing literature has discussed coupling encoded qubits with noisy physical qubits in \cite{bultrini2023battle, koukoulekidis2023framework}  and coupling logical qubits with vastly different error-rates in \cite{araki2025space, strikis2023quantum, breuckmann2017hyperbolic}.} 

\new{In a hybrid-qubit setting, it may be possible to encode a few, high-quality ancillary subsystems with a smaller code distance, optimizing the code for the physical properties. This will allow the ancilla qubits to be noise-resilient at the cost of a relatively small qubit overhead. It is important to note that, there is a {technical} violation of fault-tolerance in this case, since a covering set of physical multi-qubit gates from a smaller number of qubits to a larger cannot be transversal~\cite{huang2021between}. As such, discarding fault-tolerance is {necessary} to gain this advantage.} 

\section{Discussions}
In this work, we proposed an information-theoretic machinery called a quantum filter that can be used to purify and correct noisy quantum channels. 
We provided explicit constructions of Pauli filters which exploit the commutation structures between Pauli and Clifford gates to achieve this. We also provide a general construction that can realise quantum filters as superchannels. 
This quantum filter can reduce the errors in a noisy quantum channel and can be regarded as an intermediate form of QEM and QEC. 
Under the assumptions of clean ancillary qubits and ancillary operations (similar to \cite{lee2023error, liu2024virtual,yang2024quantum, miguel2023superposed}), our filtering scheme can fully recover the ideal Clifford channel, not only the expectation values, \new{and the exponential sampling overhead in QEM is addressed.
}
Finally, we proposed an ancilla-efficient Pauli filter which can detect and remove nearly all the low-weight Pauli errors in a noise channel with a theoretical guarantee using up to two ancillary qubits.
In cases where the noise is biased, the filtering operations can be simplified. The Pauli filter can be applied to convert unbiased noise into fully biased noise. \change{While realising these conceptual advancements, we retain many flexible properties relevant to the experimentalist. In particular, it is possible to run our scheme in both error correction and detection mode. Since the purification is done at the level of individual gates, the experimentalist can choose which gates to focus on, if the number of extra qubits is limited. As such there is sufficient opportunity to tune the filter's operations based on the available resources.}

\change{We showed that the fidelity of the to-be-purified channel can be improved with an increasing number of ancillary qubits. In contrast, this is hard to achieve with the existing QEM protocol, as its elementary block is a query to the whole channel (see the general framework of QEM~\cite{takagi2022fundamental}). The idea of adding entangled operations could stimulate further follow-ups in designing error-reduction protocols and studying the transitions between QEM and QEC. 
}
We further analysed the effect of ancilla noise on the fidelity of the corrected channel and presented the detailed analysis in Appendix \ref{appendix:error-analysis}.
We performed numerical simulations to investigate the performance of our scheme with noisy ancillary qubits, and find conditions where our scheme can still show modest gains in fidelity.  We then provided different interpretations of our scheme and compared it with other QEC codes, such as flag, coherent parity check, and space-time codes \change{and compared our scheme with flag-QEC numerically in Appendix \ref{sec:internal-noise}. 
The Pauli filter protocol can be applied as an alternative to QEM in the near-term or early fault-tolerant quantum computation~\cite{kim2023evidence,guo2024experimental,kim_scalable_2023}, with the advantage of recovering the quantum state, not just the expectation values. Thus, our protocol can be used to estimate non-linear properties of states, such as higher-order moments and entropy.}

In the specific schemes outlined in the work, we have focused on designing quantum filters for Clifford gates because it allows us to exploit the underlying commutation structure of Pauli operators. Nonetheless, a general construction of a quantum filter does not have this restriction.
In addition, from a practical standpoint, Clifford gates are easier to implement in a fault-tolerant manner in error-corrected quantum computation. The fact that we can perform only a partial purification of non-Clifford gates for general types of noise aligns with our expectation. { It is worth noting that we can convert the noise associated with non-Clifford gates into a fully biased one. This also implies that in the case of fully biased noise, we may be able to completely remove this noise in general quantum circuits.  
It is an interesting direction to explore how to extend the filter formalism beyond commutation structures, and find ways to completely purify non-Clifford resources.
We might also explore the possibilities of joint quantum channel correction and state correction schemes, where we encode both states and channels to reduce errors further. 

\new{
Since logical errors will remain non-negligible until qubit resources become sufficiently abundant and inexpensive, additional low-cost error-reduction protocols will be required to operate alongside encoded qubits in order to suppress error rates as much as possible, as emphasised in \cite{zimboras2025myths, aharonov2025importance}. Therefore, efficient error-reduction protocols with sparse interactions and low qubit overhead, such as ours, will remain highly relevant as early fault-tolerant devices become available. Protocols like the quantum filter, which can synergistically integrate the respective strengths of QEC and QEM, will play an important role in the transition toward algorithmic error reduction.
} 

\begin{acknowledgments}
J.S. and B.K. would like to thank Zhenyu Cai and Zhenhuan Liu for the early discussions related to this project. S.D. and M.K. thank David Jennings and Samson Wang for helpful discussions. We thank the UK EPSRC for their financial support through EP/T0010621/1, EP/Y004752/1,  EP/W032643/1, and EP/W003463/1. We also thank the Samsung GRC Grant and the KIST Open Innovation Grant. 
J.S. thanks the Schmidt AI in Science Fellowship supported by Schmidt Sciences LLC.
B.K. thanks the University of Oxford for a Glasstone Research Fellowship, Lady Margaret Hall, Oxford for a Research Fellowship and thanks the UK EPSRC for financial support through EP/Y004655/1 (SEEQA) and UKRI for the Future Leaders Fellowship MR/Y015843/1.
\end{acknowledgments}

\widetext

\appendix

\section{Proof of the channel correction property}
\label{app:proof-channel-correction}

Here, we provide a simple proof of the channel correction property mentioned at the end of \autoref{sec:quantum-channel-correction}, that is, a channel correction scheme that can correct Pauli errors on a single qubit is sufficient for correcting an arbitrary error. The proof is a re-packaging of the standard result from QEC, which states that if a QEC code can correct $Z$-errors and $X$-errors, it can correct arbitrary errors. Nevertheless, it is illustrative to go through this proof in the framework that we have developed. 

\begin{figure}[!h]
    \begin{center}
    \begin{tikzpicture}
    \begin{yquant}
    qubit {$\ket{0}$} q1;
    qubit {$\ket{0}$} q2;
    qubit {$\ket{\psi}$} r;
    h q1, q2;
    box {$X$} r|q1; 
    box {$Z$} r|q2; 
    [shape=yquant-rectangle, rounded corners=.3em, fill=red!20, x radius= 0.3cm]box {$P_u$} r; 
    
    box {$Z$} r|q2 ;  
    box {$X$} r|q1;
    h q1, q2;
    
    hspace {0.5cm} r, q1, q2;
    output {$\ket{\tt out}$} (r, q1, q2);
    \end{yquant}
    \end{tikzpicture}
    \end{center}
    \caption{An example of quantum filter for channel correction}
    \label{fig:single-qubit-channel-correction}
\end{figure}
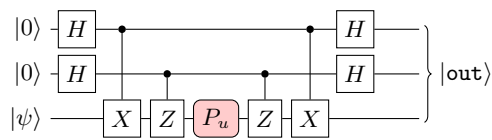

Consider the circuit in \autoref{fig:single-qubit-channel-correction}, where a Pauli error happens in the middle of the filter. For convenience, we define $(P_{00}, P_{01}, P_{10}, P_{11}) = (I, Z, X, Y)$. We have already shown in \autoref{sec:quantum-channel-correction} that, the output state $\ket{\tt out}$ of this circuit is - 
\begin{equation}\label{eq:1-qubit-full-filter}
    \ket{\psi} \to \ket{\tt out} = \ket{u} \otimes P_u\ket{\psi}
\end{equation}
where $u\in (\mathbb{Z}_2)^2$. Now, assume that an error $\ch{E}$ happens in the filter which is described by Kraus elements $\{E_k\}$, i.e, $\ch{E}(\rho) = \sum_k E_k\rho E_k$. Since the Pauli operators form a basis on the space of operators, we can write the $E_k$'s as 
\begin{equation}
    E_k = \sum_{u \in (\mathbb{Z}_2)^2} \alpha_{k, u} P_u 
\end{equation}
where the sum is over all the (single-qubit) Paulis. Hence, if $P_u$ is replaced by $E_k$ in \autoref{fig:single-qubit-channel-correction}, by linearity, the output-state will be 
\begin{equation}
    \ket{\psi} \xrightarrow{E_k} \ket{{\tt out}_k} = \sum_{u \in (\mathbb{Z}_2)^2} \alpha_{k, u} \ket{u}\otimes P_u \ket{\psi}
\end{equation}

Hence, the output state for the channel $\ch{E}$ plugged into the filter is 
\begin{align}
    \rho_{\tt out} &= \sum_k \dyad{{\tt out}_k}\\
    &= \sum_k \left(\sum_{u} \alpha_{k, u} \ket{u}\otimes P_u \ket{\psi}\right) \left(\sum_{v} \alpha_{k, v}^* \bra{v}\otimes P_v \bra{\psi}\right) \\
    &= \sum_k \sum_{u, v} \alpha_{k, u}\alpha_{k, v}^* \ketbra{u}{v} \otimes P_u (\dyad{\psi}) P_v
\end{align}

Now, in the correction phase, the ancillas are measured in the computational basis (which is equivalent to a syndrome measurement) and a corresponding Pauli is applied to the main qubit to correct for the error. As has been shown, if the measurement result on the two ancillas is $\ket{u_0}\otimes \ket{u_1} = \ket{u}$, then the corresponding correction Pauli is $P_u$. If the measurement projectors are $\Pi_u = \dyad{u}$, the correction channel is given by 
\begin{align}
    \ch{C}(\rho) = \sum_u (\mathit{\Pi}_u \otimes P_u) \rho (\mathit{\Pi}_u \otimes P_u)^{\dagger}
\end{align}
Therefore, the post-correction state is 
\begin{align}
    \ch{C}(\rho_{\tt out}) &= \sum_k \sum_{u, v, w} \alpha_{k, u}\alpha_{k, v}^* (\mathit{\Pi}_w\ketbra{u}{v}\mathit{\Pi}_w) \otimes P_w(P_u \dyad{\psi}P_v)P_w\\
    &= \left(\sum_k \sum_u \abs{\alpha_{k, u}}^2 \dyad{u}\right) \otimes \dyad{\psi}
\end{align}
Hence, the system qubit is returned to its original state and the proof is complete. The proof is easily extended for mixed input states. And for multiple qubits, one can just use a filter for each qubit separately, and repeat the procedure given above. 

\new{The proof relies on the fact that the Paulis form a basis on the Kraus-operator space. If the Kraus operators for a particular channel don't contain all the terms in the Pauli basis, then the filter construction can be fine-tuned. E.g, for the qubit dephasing channel (with strength $p$), we have,  
\begin{align*}
    E_0 = \sqrt{1-p}I, \quad  E_1 = \sqrt{p}\mathit{\Pi}_0 = \frac{\sqrt{p}}{2}(I+Z), \quad 
    E_2 = \sqrt{p}\mathit{\Pi}_1 = \frac{\sqrt{p}}{2}(I-Z)
\end{align*}
In this case, only $I$ and $Z$ appear in the expansion of the Kraus operators. As such, this can be fully corrected using only the Pauli $Z$-filter.}

\subsection{Proof of the channel correction property in terms of superchannels}
\label{app:channel-correction-superchannel}

Below, we provide an alternate proof of the channel correction property using the quantum superchannel formalism.
Consider the single-qubit channel correction filter shown in \autoref{fig:correction-superchannel}, highlighted as a superchannel. Instead of considering the correction cycle separately as in Appendix \ref{app:proof-channel-correction}, we consider the correction operations as part of the filter. 

\begin{figure}[!h]
    \centering
    \includegraphics[width=0.5\columnwidth]{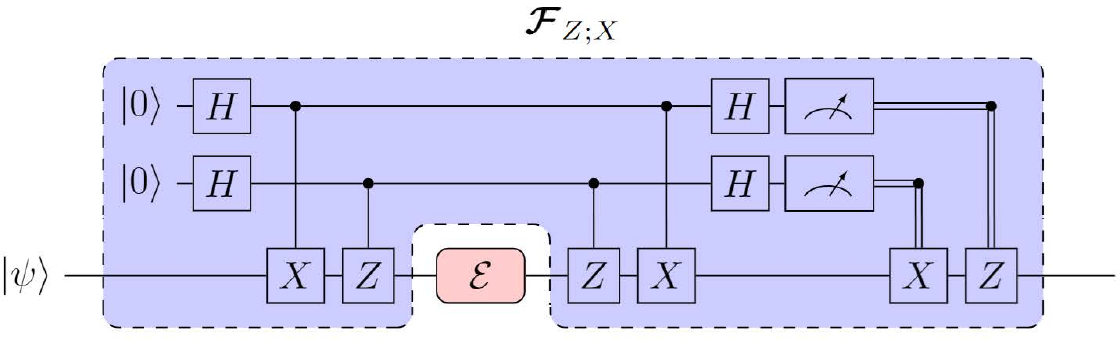}
    \caption{Single-qubit quantum channel correction filter}
    \label{fig:correction-superchannel}
\end{figure}

We use the same notation developed in \autoref{section:supermap}. The super-Kraus elements of the $Z$-correction and $X$-correction filters are 
\begin{alignat}{3}
    &\Bar{\sch{F}}_Z:\;&& \bs{F}_{Z0}[E] = I\cdot (IEI+ZEZ)/2,\; &&\bs{F}_{Z1}[E] = X\cdot (IEI-ZEZ)/2  \\
    &\sch{\Bar{F}}_X:\;&& \bs{F}_{X0}[E] = I\cdot (IEI+XEX)/2,\; &&\bs{F}_{X1}[E] = Z\cdot (IEI-XEX)/2 
\end{alignat}
The combined correction-filter is $\sch{F}_X\circ\sch{F}_Z$ whose super-Kraus elements are 
\begin{equation}
    \sch{F}_X\circ\sch{F}_Z : \bs{F}_{u}[E] = \bs{F}_{Xu_0}[\bs{F}_{Zu_1}[E]]
\end{equation}

One can check that, 
\begin{equation}
    \bs{F}_u[P_v] = \delta_{\bs{u}, \bs{v}}I
\end{equation}
for ${u}, {v}\in (\mathbb{Z}_2)^2$. Hence, $\bs{F}_u[E]\propto I $ for all the $u$ and arbitrary $E$. Thus, $\sch{F}[\ch{E}] = \ch{I}$ for arbitrary channels $\ch{E}$. Although we have shown this result for a single-qubit, it is straightforward to extend it to the case of multiple qubits. This proves that the channel correction filter can indeed correct arbitrary channels. 

\section{Partial Purification of {\it CCZ} Gates}
\label{app:partial-ccz}

\begin{figure*}[!h]

\begin{center}



















\includegraphics[width=\textwidth]{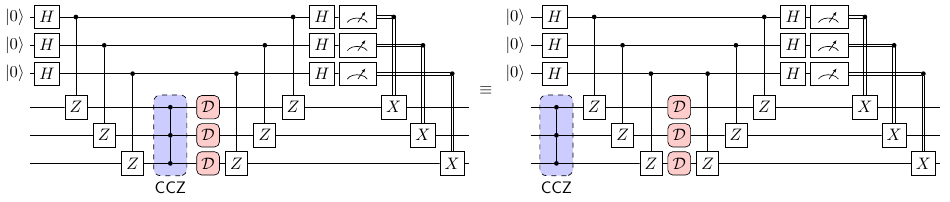}
\end{center}
\caption{Pauli Filter for partially purifying a CCZ gate. (Left) Noisy CCZ gate, modelled as an ideal CCZ gate followed by partial depolarising noise on all the qubits, plugged into a Pauli $Z$ filter on 3 qubits. (Right) Since CCZ commutes with $Z$'s, we can pull the ideal gate out of the filter. The residual depolarising noise sees the $Z$-filter, which filters out the $X$-component of the noise.
}
\label{fig:partial_filter_ccz}

\end{figure*}
Since the CCZ gate commutes with $Z\otimes Z\otimes Z$, we can adapt the scheme of \autoref{sec:partial-non-clifford} to partially purify them. Consider the circuit in \autoref{fig:partial_filter_ccz}. A noisy CCZ gate is modeled as an ideal CCZ gate followed by a local depolarising channel $\ch{D}^{\otimes 3}$, with $\ch{D}$ defined as in \autoref{eqn:general_dep_error} ($\mathcal{D} = (1-p) \ch{I} + p_X \ch{X} + p_Y\ch{Y} +  p_Z\ch{Z}$). Since the single-qubit Paulis $Z_1, Z_2, Z_3$ commute with the CCZ gate, we can pull the CCZ gate out of the filter. Then, following the analysis of \autoref{sec:partial-non-clifford}, the channel $\ch{D}^{\otimes 3}$ is transformed to the channel $\Tilde{\ch{D}}^{\otimes 3}$, with $\Tilde{\ch{D}}$ defined in \autoref{eqn:biased_dep_channel} ($\Tilde{\mathcal{D}} = (1-p_Y-p_Z)\cdot \ch{I} +  (p_Y+p_z)\cdot\ch{Z}$). For the case $p_X = p_Y = p_Z = p/3$, this results in an increase in fidelity of the CCZ gate from $(1-p)^3$ to $(1-2p/3)^3$.

\section{Analysis of the ancilla-efficient Pauli filter}
\label{app:ge-filter}
Here, we find out some properties of the ancilla-efficient Pauli filter $\sch{F}_g$ introduced in Section \ref{sec:good-enough}, particularly how many high-weight Paulis it can remove, its success probability for a local depolarising noise on $n$-qubits, and fidelity of the output channel if it succeeds. 

\subsection{Number of Paulis removed by the ancilla-efficient Pauli filter}

Suppose, a weight $w$ Pauli contains $(x, y, z)$ number of $X$'s, $Y$'s and $Z$'s, i.e, 
\begin{equation}
    w = x + y + z
\end{equation}
We will denote such a Pauli as $P_{x, y, z}$. By \autoref{eqn:ae-pauli-map}$, \sch{F}_{AE}$ can remove $P_{x, y, z}$, if and only if 
\begin{equation}
    I^{\otimes n}(P_{x, y, z})I^{\otimes n} + Z^{\otimes n}(P_{x, y, z})Z^{\otimes n} + X^{\otimes n}(P_{x, y, z})X^{\otimes n} + Y^{\otimes n}(P_{x, y, z})Y^{\otimes n} = 0
\end{equation}
Using the anti-commutation property of the Paulis, this reduces to 
\begin{align}
    &1 + (-1)^{x+y} + (-1)^{y+z} + (-1)^{(z+x)} = 0\\
    \Rightarrow & (-1)^w + (-1)^x + (-1)^y + (-1)^z = 0\label{eq:good-enough-condition}
\end{align}
 \autoref{eq:good-enough-condition} will be satisfied if 2 of the numbers $w, x, y, z$ are odd and the other 2 are even. 

First, let us consider the case where $w$ is odd. Then, among $x, y, z$ one must be odd, and the other two must be even. Let, $S_w[o, e, e]$ be the set of all such tuples $(x, y, z)$. Then, the number of such Paulis of weight-$w$ that satisfy  \autoref{eq:good-enough-condition} is 
\begin{equation}
    \#(w) = \binom{n}{w} \cdot \sum_{(x, y, z) \in S_w[o, e, e]} \binom{w}{x, y, z}
\end{equation}
Similarly, we may define a set $S_w[o, o, o]$ of tuples $(x, y, z)$ such that, they sum to $w$ and all three of them are odd. If $w$ is odd, any tuple $(x, y, z)$ of non-negative integers which sum to $w$ will belong to one of these two sets. Therefore, we have, 
\begin{equation}\label{eqn:multinomial_sum}
    \sum_{(x, y, z) \in S_w[o, e, e]} \binom{w}{x, y, z} + \sum_{(x, y, z) \in S_w[o, o, o]} \binom{w}{x, y, z} = \sum_{x, y, z} \binom{w}{x, y, z} = 3^w
\end{equation}
Finding an exact expression for $\#(w)$ might be difficult in general. Therefore, we will try to bound this quantity instead.  

Consider a tuple $(x, y, z) \in S_w[o, e, e]$. Let us call the odd element of this tuple $o$, and the even elements $e_1, e_2$. We define a map $f$ which takes $o$ to $(o-2)$ and $e_1, e_2$ to $e_1 + 1, e_2 + 1$. Notice that, all 3 of these numbers are odd. Therefore, $f$ maps a tuple of $S_w[o, e, e]$ to a tuple of $S_w[o, o, o]$ (with the possibility that $(o-2)$ might be negative). Then, if we define $\binom{w}{x, y, z}$ to be 0 if any one of $x, y, z$ is negative, we can write 
\begin{align}
    \sum_{(x, y, z) \in S_w[o, o, o]} \binom{w}{x, y, z} &= \sum_{(x, y, z) \in S_w[o, e, e]} \binom{w}{f[(x, y, z)]}\\
    &= \sum_{(x, y, z) \in S_w[o, e, e]} \left[\frac{o(o-1)}{(e_1+1)(e_2+1)}\right]_{(x, y, z)}\binom{w}{x, y, z}\\
    &\leq \left(\max_{(x, y, z) \in S_w[o, e, e]} \frac{o(o-1)}{(e_1+1)(e_2+1)}\right) \sum_{(x, y, z) \in S_w[o, e, e]} \binom{w}{x, y, z}\\
    &= w(w-1) \sum_{(x, y, z) \in S_w[o, e, e]} \binom{w}{x, y, z}\label{eqn:ae-main-inequality}
\end{align}
where in the second line, $(o, e_1, e_2)_{(x, y, z)}$ are the odd and even elements of the tuple $(x, y, z)$ respectively. 

Using this, we have, 
\begin{equation}
    3^w = \sum_{(x, y, z) \in S_w[o, e, e]} \binom{w}{x, y, z} + \sum_{(x, y, z) \in S_w[o, o, o]} \binom{w}{x, y, z} \leq (w^2 - w + 1)\sum_{(x, y, z) \in S_w[o, e, e]} \binom{w}{x, y, z}
\end{equation}

Plugging this in  \autoref{eqn:multinomial_sum} gives us the bound, 
\begin{equation}
    \#(w) \geq \binom{n}{w} \frac{3^w}{1+w(w-1)}
\end{equation}
This argument can be translated verbatim for the case where $w$ is even. 

Therefore, the total number of components that are removed by $\sch{F}_{AE}$ satisfies 
\begin{equation}
    \sum_{w=1}^n \#(w) \geq \sum_{w=1}^n \binom{n}{w} \frac{3^w}{1+w(w-1)}
\end{equation}
As a conservative estimate, we may bound the terms $1/(1+w(w-1))$ by their maximum value $1/(1+n(n-1))$. This gives us the bound, 
\begin{align*}
    \sum_{w=1}^n \#(w) \geq \frac{1}{1+n(n-1)} \sum_{w=1}^n \binom{n}{w}3^w 
    = \frac{4^n}{n^2-n+1}  
    \sim \Omega(4^n/n^2)
\end{align*}
\paragraph{Average weight of removed Paulis} The average weight of the Paulis that are removed is bounded by 
\begin{align*}
    \Bar{w} = \frac{\sum_{w=1}^n w\cdot \#(w)}{\sum_{w=1}^n \#(w)} \geq \frac{\sum_{w=1}^n \frac{1}{w}\binom{n}{w}}{\sum_{w=1}^n \frac{1}{w^2}\binom{n}{w}}
\end{align*}
Asymptotically, the dominant term in these sums is $\binom{n}{n/2}$. Hence, in the large-$n$ limit, we have, 
\begin{align*}
    \Bar{w} \sim \frac{\frac{1}{n/2}\binom{n}{n/2}}{\frac{1}{(n/2)^2}\binom{n}{n/2}} = n/2
\end{align*}

\subsection{Fidelity of the purified channel}
Throughout this treatment, we will use the term fidelity to mean the entanglement fidelity of a channel, or the co-efficient of the identity element of its Pauli transfer matrix.

First, we find the probability that our protocol succeeds. The number of Paulis that are {not} removed by $\sch{F}_{AE}$ is 
\begin{equation}
    \#'(w) := \binom{n}{w}3^w - \#(w) \leq \binom{n}{w}3^w\left(1-\frac{1}{1+w(w-1)}\right) 
\end{equation}
Since $w \leq n$, some simple algebra gives us the bound, 
\begin{equation}
    \#'(w) \leq n\cdot\frac{n^2-1}{n^3+1} \binom{n}{w} 3^w
\end{equation}
The probability that a particular Pauli of weight $w$ occurs is $$(p/3)^wq^{n-w}$$
where $q=1-p$. 
Therefore, the probability of success is 
\begin{align}
    p_s &\leq q^n + n\cdot\frac{n^2-1}{n^3+1} \sum_{w=1}^n \binom{n}{w} 3^w (p/3)^w q^{n-w}\\
    &= q^n + n\cdot\frac{n^2-1}{n^3+1} (1-q^n)\\
    &= 1 - (1-q^n)\frac{1+n}{1+n^3}\label{eqn:success_probability_final}
\end{align}
Therefore, the fidelity of the purified output channel is 
\begin{align*}
    F_o = q^n/p_s
    \geq \frac{q^n}{1 - (1-q^n)\frac{1+n}{1+n^3}}
    \geq q^n \left[1 + (1-q^n)\frac{1+n}{1+n^3}\right]
    \geq q^n \left[1 + (1-q^n)\frac{1}{n^2}\right]
\end{align*}
In the limit that $n$ is large (and $q^n$, the input fidelity, is vanishingly small), the rate of increase is approximately, 
\begin{equation}
    F_o/F_i \simeq 1 + \frac{1}{n^2}
\end{equation}

\subsection{Proof of quadratic reduction in infidelity}
\label{app:quadratic_reduction}

Let, the infidelity of the input channel be $\epsilon$, i.e, 
\begin{equation}
    F_i = q^n := 1 - \epsilon
\end{equation}
Then, the number of {channel} qubits in terms of the infidelity is 
\begin{equation}\label{eqn:def_n}
    n = \frac{\ln{(1-\epsilon)}}{\ln{q}} = \frac{\ln(1-\epsilon)}{\ln(1-p)}
\end{equation}
Using the relation, 
\begin{equation}
    \epsilon \leq -\ln(1-\epsilon) < \frac{\epsilon}{1-\epsilon} \qquad \epsilon < 1
\end{equation}
we have, 
\begin{equation}
    \frac{\epsilon/p}{1-p} \leq n \leq \frac{\epsilon/p}{1-\epsilon}
\end{equation}

Now, the output fidelity of the channel is bounded by 
\begin{align}\label{eqn:conservative_infidelity}
    F_o \geq \frac{q^n}{1-npq^{n-1}}
\end{align}
where we have considered removing only the weight-1 Paulis. Substituting $q^n = 1-\epsilon$ and $n\geq \epsilon/pq$, we get , 
\begin{align}
    F_o &\geq \frac{1-\epsilon}{1- \epsilon(1-\epsilon)/q^2}&&\\
    &\geq \frac{1-\epsilon}{1- \epsilon(1-\epsilon)}  &[\because q^2 < 1]&\\
    &\geq (1-\epsilon)(1+ \epsilon(1-\epsilon)) &[\because \epsilon(1-\epsilon) < 1]&\\
    &\geq 1 - 2\epsilon^2
\end{align}
Hence, up to leading order, the output infidelity is $O(\epsilon^2)$, which is a {\bf quadratic} reduction in the infidelity.

In fact,  \autoref{eqn:conservative_infidelity} is a conservative estimate of the output fidelity. A better asymptotic (i.e, large $n$) estimate can be found by using  \autoref{eqn:success_probability_final}. In terms of the infidelity, we have, 
\begin{align*}
    p_s  \leq 1 - (1-q^n)\frac{1}{n^2}
     \leq 1 - \epsilon\cdot \frac{p^2}{\epsilon^2}
     \leq 1 - p^2/\epsilon
\end{align*}
Using this bound, the output fidelity becomes, 
\begin{align}\label{eq:good-infidelity-bound}
    F_o &\geq \frac{1-\epsilon}{1-p^2/\epsilon}
\end{align}
which gives a better estimate in the regime, 
\begin{equation}\label{eqn:ge_condition}
    p^2/\epsilon > \epsilon(1-\epsilon) 
\end{equation}
Note that
\begin{equation}
    \epsilon^2(1-\epsilon) \leq (1-\epsilon) = q^n = e^{n \ln{q}} \leq e^{-np}
\end{equation}
Hence, a sufficient condition for \autoref{eq:good-infidelity-bound} to give a better bound on the output fidelity is: 
\begin{equation}
    p^2 > e^{-np} \Rightarrow n > 2(1/p)\ln(1/p)
\end{equation}

\subsection{A Global Pauli Noise Model}
\label{section:global-pauli-noise-model}
In this section, we consider the performance of the ancilla-efficient filter for an exceptionally decaying global Pauli noise-model on $n$-qubits. Suppose, the probability that a weight-$w$ Pauli error occurs is 
\begin{equation}
    p(w) \propto \epsilon^w
\end{equation}
for some $\epsilon < 1$. The constant of proportionality is $F = (1-\epsilon)/(1-\epsilon^n)$ which is found from normalization. This is also equal to the fidelity of the channel, which is the probability of no errors occuring ($p(0)$). Hence, the infidelity of the channel is, 
\begin{equation}
    1 - F = 1 - \frac{1-\epsilon}{1 - \epsilon^n} = \epsilon - O(\epsilon^n)
\end{equation}
which, up to leading order is $\epsilon$. 
The probability measure on the weight-1 Paulis is $F\epsilon$. Hence, if we consider removing these using the ancilla-efficient filter, the output fidelity of the channel is 
\begin{align}
    F_{\rm out} &\geq \frac{F}{1-F\epsilon}
    = \frac{1-\epsilon}{1-\epsilon^n - (1-\epsilon)\epsilon}
     \geq \frac{1-\epsilon}{1-(1-\epsilon)\epsilon} \geq 1-2\epsilon^2
\end{align}
Hence, the output infidelity is again $O(\epsilon^2)$. We see that, the quadratic reduction in infidelity also holds up in the case of this model of decaying global noise. 

\section{Analysis of ancilla noise}\label{appendix:error-analysis}

So far, we have considered the scenario that the filter operations are noiseless. However, it would be instructive to have an estimate of the performance of the filter if the filter operations introduce errors themselves. In order to do this, we consider the single-qubit Pauli correction filter, shown in \autoref{fig:channel-correction-modified}, with measurements moved at the end of the circuit for the ease of the analysis that is about to follow.  
\begin{figure}[!h]
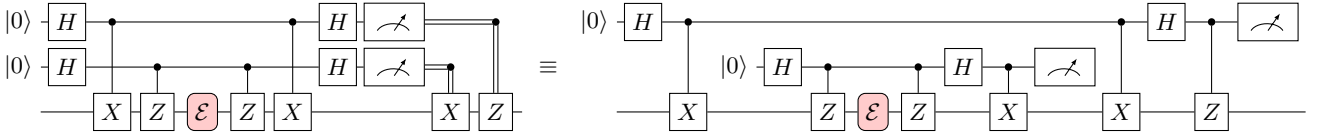

    \begin{yquantgroup}
        \registers{
        qubit {} q1; 
        qubit {} q2; 
        qubit {} r; 
        }
        \circuit{
        init {$\ket{0}$} q1;
        init {$\ket{0}$} q2;
        h q1, q2;
        box {$X$} r|q1; 
        box {$Z$} r|q2; 
        [shape=yquant-rectangle, rounded corners=.3em, fill=red!20]box {$\mathcal{E}$} r; 
        
        box {$Z$} r|q2 ;  
        box {$X$} r|q1;
        h q1, q2;

        measure q1;
        measure q2;
        
        box {$X$} r|q2;
        box {$Z$} r|q1;
        
        discard q1; 
        discard q2;
        }
        \equals[$\equiv$]
        \circuit{
        discard q2; 
        init {$\ket{0}$} q1;
        
        h q1;
        box {$X$} r|q1; 

        init {$\ket{0}$} q2;
        h q2; 
        box {$Z$} r|q2; 
        [shape=yquant-rectangle, rounded corners=.3em, fill=red!20]box {$\mathcal{E}$} r; 
        
        box {$Z$} r|q2 ;
        h q2; 
        box {$X$} r|q2; 
        measure q2;
        discard q2; 
        
        box {$X$} r|q1;
        h q1;
        box {$Z$} r|q1;
        measure q1;
        
        discard q1; 
        }
    \end{yquantgroup}
    \caption{single-qubit Pauli correction filter (with measurements pushed to the end).}
    \label{fig:channel-correction-modified}
\end{figure}

We assume noisy 2-qubit $CZ$ and $CX$ gates, with local depolarising noise on the control and target qubits after each application of the gate. 
\begin{equation}
    (\widetilde{CZ}/\widetilde{CX}) = (\ch{D}_{p_c}\otimes \ch{D}_{p_t}) \circ (CZ/CX)
\end{equation}

\begin{figure}[!h]
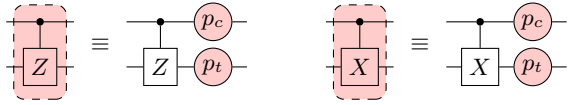

    \begin{center}
    \begin{yquantgroup}
    \registers{
    qubit {} q1;
    qubit {} q2;
    } 
    \circuit{
        [this subcircuit box style={draw, dashed, rounded corners,
        fill=red!20, inner ysep=5pt, inner xsep=0pt},
        register/default name=]
        subcircuit {
        qubit {} q[2]; 
        box {$Z$} q[1]|q[0]; 
        } (q1, q2);
    }
\equals[$\equiv$]
\circuit{ 
\yquantset{operator/separation=2mm}
        box {$Z$} q2|q1 ; 
        [shape=yquant-circle, radius = 2.5mm, fill=red!20]box {$p_c$} q1; 
        [shape=yquant-circle, radius = 2.5mm, fill=red!20]box {$p_t$} q2; 
    }
\equals[$\qquad$]
\circuit{    
        [this subcircuit box style={draw, dashed, rounded corners,
        fill=red!20, inner ysep=5pt, inner xsep=0pt},
        register/default name=]
        subcircuit {
        qubit {} q[2]; 
        box {$X$} q[1]|q[0]; 
        } (q1, q2);
    }
\equals[$\equiv$]
\circuit{ 
\yquantset{operator/separation=2mm}
        box {$X$} q2|q1 ; 
        [shape=yquant-circle, radius = 2.5mm, fill=red!20]box {$p_c$} q1; 
        [shape=yquant-circle, radius = 2.5mm, fill=red!20]box {$p_t$} q2; 
    }
\end{yquantgroup}

    \end{center}
    \caption{Noisy $CZ$ and $CX$ gates. Each noisy 2-qubit gate is modelled as an ideal gate, followed by local depolarising error on the control and target qubits (with probability $p_c$ and $p_t$ respectively).}
    \label{fig:noisy-gates}
\end{figure}

To find the effect of these errors, we use the (Pauli) error propagation rules for the $CZ$ and $CX$ gates outlined in \autoref{fig:pauli-propagation-rules}. 

\begin{figure*}[!h]
    \includegraphics[width=0.85\columnwidth]{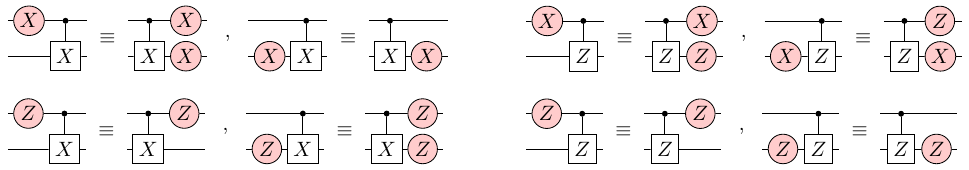}
    \caption{Pauli error propagation rules for the $CX$ (left) and $CZ$ (right) gates.}
    \label{fig:pauli-propagation-rules}
\end{figure*}

Consider the filtering unitary $F_Z$, outlined in \autoref{fig:erroneous-circuit}, where we have marked locations of possible (Pauli) faults. Then, we denote by $F_Z[U|P_A]$ by the unitary that results if we replace the location $A$ by a Pauli $P_A$ (which can be either $I, X, Y, Z$). We denote $F_Z[I_A]$ simply by $F_Z$, which is the unitary without any faults. Using the propagation rules, we can propagate this Pauli error $P_A$ forward in time through all the gates that come after it. This results in an expression for what errors remain on the final state after the circuit is executed with a fault at location $A$. For example, one can check that, 
\begin{equation}
    \bs{F_Z}[U|X_A] = Z_1Z_2\cdot \bs{F_Z}[U]
\end{equation}
Therefore, as the error $X_A$ propagates through the circuit, it results in $Z$ errors on both the qubits (on top of the final state output from the filter). 
In this way, we might find out how $Z$ and $X$ faults at the sites $A, C, D$ propagate to the end of the filter. We can skip the calculations for the site $C$, since this may be considered as part of the error channel $\mathcal{E}$ that we are trying to correct. The result is shown in Table \ref{table:error-propagation}.

\begin{figure}
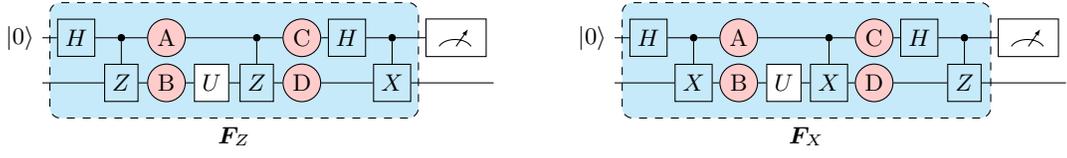

    \centering
    \begin{yquantgroup}
        \registers{
        qubit {} q; 
        qubit {} r; 
        }
        \circuit{
        init {$\ket{0}$} q; 
        
        [this subcircuit box style={draw, dashed, rounded corners,
        fill=cyan!20, inner ysep=6pt, inner xsep=0pt, "$\bs{F}_Z$" below},
        register/default name=]
        subcircuit {qubit {} q; 
        qubit {} r; 
        h q; 
        box {$Z$} r|q; 
        [shape=yquant-circle, radius = 2.5mm, fill=red!20]box {A} q; 
        [shape=yquant-circle, radius = 2.5mm, fill=red!20]box {B} r; 
        [fill=white]box {$U$} r; 
        box {$Z$} r|q; 
        [shape=yquant-circle, radius = 2.5mm, fill=red!20]box {C} q; 
        [shape=yquant-circle, radius = 2.5mm, fill=red!20]box {D} r; 
        h q; 
        box {$X$} r|q; 
        } (q, r);
        measure q;
        discard q; 
        }
        \equals[$\qquad$]
        \circuit{
        init {$\ket{0}$} q; 
        
        [this subcircuit box style={draw, dashed, rounded corners,
        fill=cyan!20, inner ysep=6pt, inner xsep=0pt, "$\bs{F}_X$" below},
        register/default name=]
        subcircuit {qubit {} q; 
        qubit {} r; 
        h q; 
        box {$X$} r|q; 
        [shape=yquant-circle, radius = 2.5mm, fill=red!20]box {A} q; 
        [shape=yquant-circle, radius = 2.5mm, fill=red!20]box {B} r; 
        [fill=white]box {$U$} r; 
        box {$X$} r|q; 
        [shape=yquant-circle, radius = 2.5mm, fill=red!20]box {C} q; 
        [shape=yquant-circle, radius = 2.5mm, fill=red!20]box {D} r; 
        h q; 
        box {$Z$} r|q; 
        } (q, r);
        measure q;
        discard q; 
        }
    \end{yquantgroup}
    
    \caption{Filter unitary for the Pauli-$Z$ (left) and Pauli-$X$ (right)  Correction Filters with possible (Pauli) faults at locations A, B, C, D.}
    \label{fig:erroneous-circuit}
\end{figure}

\begin{table}[]
\begin{tabular}{cc|c|c|c|c|c|c}
\multicolumn{2}{c|}{Error}                          & $X_A$ & $Z_A$ & $X_C$ & $Z_C$ & $X_D$ & $Z_D$ \\ \hline\hline
\multicolumn{1}{c|}{\multirow{2}{*}{$F_Z$}} & $q_1$ & $Z_1$ & $X_1$ & $Z_1$ & $X_1$ & $I_1$ & $Z_1$ \\ 
\multicolumn{1}{c|}{}                       & $q_2$ & $Z_2$ & $X_2$ & $I_2$ & $X_2$ & $X_2$ & $Z_2$ \\ \hline
\multicolumn{1}{c|}{\multirow{2}{*}{$F_X$}} & $q_1$ & $Z_1$ & $X_1$ & $Z_1$ & $X_1$ & $Z_1$ & $I_1$ \\ 
\multicolumn{1}{c|}{}                       & $q_2$ & $X_2$ & $Z_2$ & $I_2$ & $Z_2$ & $X_2$ & $Z_2$ \\ \hline
\end{tabular}
\caption{Propagation of Pauli faults through the filter circuit.}
\label{table:error-propagation}
\end{table}

Then, if we assume depolarising noise at all the fault sites and the error channel, the faulty filter is equivalent to (\autoref{fig:error-analysis}) 
\begin{equation}
    \Tilde{\sch{F}}[\ch{D}_{p_e}] = \ch{D}_{p_0}\circ\ch{D}^{(X)}_{p_X}\circ \ch{D}^{(Z)}_{p_Z}
\end{equation}
where $\ch{D}, \ch{D}^{(X)}, \ch{D}^{(Z)}$ are partially depolarising, bit-flip and phase-flip channels respectively. The error probabilities of these channels are 
\begin{align}
    p_0 &= 1 - (1-p_c)(1-p_t) \approx p_c+p_t\\
    p_X &= 1 - (1-2p_c/3)^2(1-2p_t/3)^2 \approx 2(p_c+p_t)/3\\
    p_Z &= 2p_c/3 
\end{align}
where the approximations hold for $p_c, p_t \ll 1$. Then, the fidelity of the output from the faulty filter is 
\begin{align}
    \Tilde{F} &= (1-p_0)(1-p_X)(1-p_Z) + p_0(p_X + p_Z + p_Zp_X)/3\\
    & \geq (1-p_c)(1-p_t)\cdot(1-2p_c)^2(1-2p_t)^2\cdot (1-2p_c/3) := F_{\rm critical}\label{eqn:sensititvity}
\end{align}
\autoref{eqn:sensititvity} shows that the lower bound on the output fidelity is more sensitive to $p_c$ (as it has an extra factor $(1-2p_c/3)$ of dependence on $p_c$, other than which the terms involving $p_c$ and $p_t$ are symmetrical). As a result, this provides analytical evidence for the intuition that the fidelity of the output channel is more sensitive to the error rate on the ancilla. In addition, \autoref{eqn:sensititvity} also provides an upper bound $F_{\rm critical}$ on the fidelity of the error channel $\ch{E}$, below which, using the filter provides an advantage over using the unfiltered circuit. 


\begin{figure*}[!h]
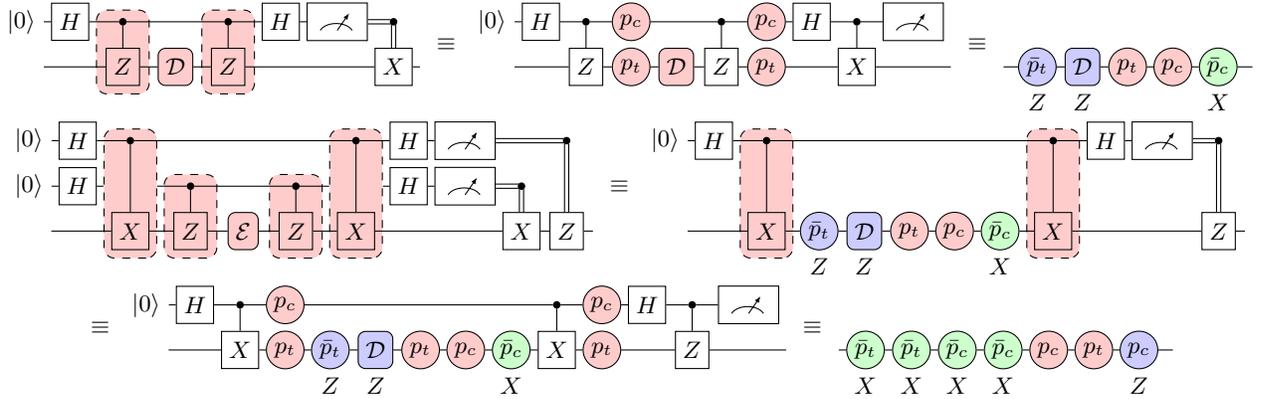

\begin{yquantgroup}
    \registers{
    qubit {} q2;
    qubit {} r;
    }
    \circuit{
    init {$\ket{0}$} q2; 
    h q2;
    [this subcircuit box style={draw, dashed, rounded corners,
        fill=red!20, inner ysep=3pt, inner xsep=0pt},
        register/default name=]
        subcircuit {
        qubit {} q[2]; 
        box {$Z$} q[1]|q[0]; 
        } (q2, r);
    [shape=yquant-rectangle, rounded corners=.3em, fill=red!20]box {$\mathcal{D}$} r; 
    [this subcircuit box style={draw, dashed, rounded corners,
        fill=red!20, inner ysep=3pt, inner xsep=0pt},
        register/default name=]
        subcircuit {
        qubit {} q[2]; 
        box {$Z$} q[1]|q[0]; 
        } (q2, r);
    h q2;
    measure q2;
    box {$X$} r|q2;
    discard q2;
    }
    \equals[$\equiv$]
    \circuit{
    init {$\ket{0}$} q2;
    h q2;
    
    box {$Z$} r|q2; 
    [shape=yquant-circle, radius = 2.5mm, fill=red!20]box {$p_{c}$} q2; 
    [shape=yquant-circle, radius = 2.5mm, fill=red!20]box {$p_{t}$} r;
    
    [shape=yquant-rectangle, rounded corners=.3em, fill=red!20]box {$\mathcal{D}$} r; 
    
    box {$Z$} r|q2 ;
    [shape=yquant-circle, radius = 2.5mm, fill=red!20]box {$p_{c}$} q2; 
    [shape=yquant-circle, radius = 2.5mm, fill=red!20]box {$p_{t}$} r;
    
    h q2;
    
    box {$X$} r|q2;
    
    measure q2;
    discard q2;
    }
    \equals[$\equiv$]
    \circuit{
    discard q2; 
    hspace {1mm} r;
    [shape=yquant-circle, radius = 2.5mm, fill=blue!20, "$Z$" below]box {$\bar{p}_{t}$} r; 
    [shape=yquant-rectangle, rounded corners=.3em, fill=blue!20, "$Z$" below]box {$\mathcal{D}$} r; 
    [shape=yquant-circle, radius = 2.5mm, fill=red!20]box {${p}_{t}$} r; 
    [shape=yquant-circle, radius = 2.5mm, fill=red!20]box {${p}_{c}$} r; 
    [shape=yquant-circle, radius = 2.5mm, fill=green!20, "$X$" below]box {$\bar{p}_{c}$} r; 
    hspace {1mm} r;
    }
\end{yquantgroup}\\
\begin{yquantgroup}
    \registers{
    qubit {} q1;
    qubit {} q2;
    qubit {} r;
    }
    \circuit{
    init {$\ket{0}$} q1; 
    init {$\ket{0}$} q2; 
    h q1, q2; 
    [this subcircuit box style={draw, dashed, rounded corners,
        fill=red!20, inner ysep=3pt, inner xsep=0pt},
        register/default name=]
        subcircuit {
        qubit {} q[2]; 
        box {$X$} q[1]|q[0]; 
        } (q1, r);
        
    [this subcircuit box style={draw, dashed, rounded corners,
        fill=red!20, inner ysep=3pt, inner xsep=0pt},
        register/default name=]
        subcircuit {
        qubit {} q[2]; 
        box {$Z$} q[1]|q[0]; 
        } (q2, r);
    [shape=yquant-rectangle, rounded corners=.3em, fill=red!20]box {$\mathcal{E}$} r; 

    [this subcircuit box style={draw, dashed, rounded corners,
        fill=red!20, inner ysep=3pt, inner xsep=0pt},
        register/default name=]
        subcircuit {
        qubit {} q[2]; 
        box {$Z$} q[1]|q[0]; 
        } (q2, r);
    
    [this subcircuit box style={draw, dashed, rounded corners,
        fill=red!20, inner ysep=3pt, inner xsep=0pt},
        register/default name=]
        subcircuit {
        qubit {} q[2]; 
        box {$X$} q[1]|q[0]; 
        } (q1, r);

    h q1, q2;

    measure q1;
    measure q2;
    
    box {$X$} r|q2;
    box {$Z$} r|q1;

    discard q1; 
    discard q2;
    }
    \equals[$\equiv$]
    \circuit{
    discard q2;
    init {$\ket{0}$} q1; 
    
    h q1;

    [this subcircuit box style={draw, dashed, rounded corners,
        fill=red!20, inner ysep=3pt, inner xsep=0pt},
        register/default name=]
        subcircuit {
        qubit {} q[2]; 
        box {$X$} q[1]|q[0]; 
        } (q1, r);
    
    
    [shape=yquant-circle, radius = 2.5mm, fill=blue!20, "$Z$" below]box {$\bar{p}_{t}$} r; 
    [shape=yquant-rectangle, rounded corners=.3em, fill=blue!20, "$Z$" below]box {$\mathcal{D}$} r; 
    [shape=yquant-circle, radius = 2.5mm, fill=red!20]box {${p}_{t}$} r; 
    [shape=yquant-circle, radius = 2.5mm, fill=red!20]box {${p}_{c}$} r; 
    [shape=yquant-circle, radius = 2.5mm, fill=green!20, "$X$" below]box {$\bar{p}_{c}$} r; 

    [this subcircuit box style={draw, dashed, rounded corners,
        fill=red!20, inner ysep=3pt, inner xsep=0pt},
        register/default name=]
        subcircuit {
        qubit {} q[2]; 
        box {$X$} q[1]|q[0]; 
        } (q1, r);
    

    h q1; 
    measure q1;
    box {$Z$} r|q1; 
    discard q1;
    }
\end{yquantgroup}\\
\begin{yquantgroup}
    \registers{
    qubit {} q1;
    qubit {} r;
    }
    \equals[$\equiv$]
    \circuit{
    init {$\ket{0}$} q1; 
    h q1; 
    box {$X$} r|q1; 
    [shape=yquant-circle, radius = 2.5mm, fill=red!20]box {$p_{c}$} q1; 
    [shape=yquant-circle, radius = 2.5mm, fill=red!20]box {$p_{t}$} r;
    
    [shape=yquant-circle, radius = 2.5mm, fill=blue!20, "$Z$" below]box {$\bar{p}_{t}$} r; 
    [shape=yquant-rectangle, rounded corners=.3em, fill=blue!20, "$Z$" below]box {$\mathcal{D}$} r; 
    [shape=yquant-circle, radius = 2.5mm, fill=red!20]box {${p}_{t}$} r; 
    [shape=yquant-circle, radius = 2.5mm, fill=red!20]box {${p}_{c}$} r; 
    [shape=yquant-circle, radius = 2.5mm, fill=green!20, "$X$" below]box {$\bar{p}_{c}$} r; 

    box {$X$} r|q1; 
    [shape=yquant-circle, radius = 2.5mm, fill=red!20]box {$p_{c}$} q1; 
    [shape=yquant-circle, radius = 2.5mm, fill=red!20]box {$p_{t}$} r;

    h q1; 
    box {$Z$} r|q1; 
    measure q1; 
    discard q1; 
    }
    \equals[$\equiv$]
    \circuit{
    discard q1; 

    [shape=yquant-circle, radius = 2.5mm, fill=green!20, "$X$" below]box {$\bar{p}_{t}$} r; 
    [shape=yquant-circle, radius = 2.5mm, fill=green!20, "$X$" below]box {$\bar{p}_{t}$} r; 
    [shape=yquant-circle, radius = 2.5mm, fill=green!20, "$X$" below]box {$\bar{p}_{c}$} r; 
    [shape=yquant-circle, radius = 2.5mm, fill=green!20, "$X$" below]box {$\bar{p}_{c}$} r;
    [shape=yquant-circle, radius = 2.5mm, fill=red!20]box {$p_{c}$} r;
    [shape=yquant-circle, radius = 2.5mm, fill=red!20]box {$p_{t}$} r;
    [shape=yquant-circle, radius = 2.5mm, fill=blue!20, "$Z$" below]box {$p_{c}$} r;
    hspace {1mm} r;
    }
\end{yquantgroup}

\caption{Error propagation through the channel correction filter. Red, blue and green circles represent partially depolarising, phase-flip and bit-flip channels respectively, with corresponding error probabilities written inside the circles. $\Bar{p}_c = 2p_c/3, \Bar{p}_t = 2p_t/3$.}
\label{fig:error-analysis}
\end{figure*}

\section{Connections and differences to related schemes}
\label{appendix:connections}
In the discussions, we touched upon how one could draw connections between our scheme and other error-correction/reduction procedures. Here, we expand on these themes, and detail how our scheme significantly differs from them despite the similarities in sub-circuit-structures. We focus on Pauli Twirling, the Iceberg Code and Entanglement-Assisted QEC (EAQEC) codes, and the Flag-QEC scheme. 

\subsection{Pauli Twirling}
{Consider a noisy $n$-qubit unitary channel $\ch{U}_{\ch{E}} = \ch{E} \circ \ch{U}$, with the ideal channel being $\ch{U}(\rho) = U\rho U^{\dagger}$. Pauli-Twirling this channel might be written as the following map, 
\begin{equation}
    \sch{T}[\ch{U}_{\ch{E}}] = \frac{1}{4^n}\sum_P \ch{P}'\circ \ch{U}_{\ch{E}} \circ \ch{P} =\sch{T}[\ch{E}] \circ \ch{U}
\end{equation}
where, the sum is over all $n$-qubit Paulis $P$. Here, $\ch{P}(\rho) = P\rho P, \ch{P}'(\rho) = P'\rho P'$, the operator $P'$ being the conjugate of $P$ with respect to $U$, i.e,  
$$P' = UPU^{\dagger}$$
Thus, in order to perform Pauli-Twirling efficiently, we focus on unitaries for which the $P'$ are efficiently computable and implementable. The canonical example is to focus on Clifford channels, for which the $P'$ are Paulis as well. This may indicate a similarity to the Clifford channel purification scheme outlined in \autoref{sec:clifford-purification}, where, the $\textsc{Select}$ operations are constructed out of conjugate Paulis with respect to a Clifford as well.} However, this is the only structural similarity, and there are two main points of divergence:
\begin{itemize}
    \item
    In our filtering scheme, measurement results allow us to increase the channel fidelity by post-selection or correction. In a QEC setting, Pauli twirling is known to be equivalent to simply discarding the results of measuring a set of stabiliser operators \cite{cai2019constructing}, and does not increase the channel fidelity~\cite{nielsen2002simple}\footnote{Pauli twirling collapses the Pauli transfer matrix of the process so that it is diagonal. Considering only a correctable subset of errors, when we have an encoded state in QEC that is able to distinguish such Pauli errors this is identical to the average of the results we would obtain if we measured all the stabilisers of the state.}.    

    To be more concrete, if a single-qubit channel $\ch{E}$ has process matrix $\chi_{_{P, Q}}$, 
    \begin{equation}
        \ch{E}(\rho) = \sum_{P, Q}\chi_{_{P, Q}} P\rho Q
    \end{equation}
    where $P, Q$ are Paulis, then, if we apply the $Z-$Filter on $\ch{E}$, and post-select on the ancilla being $\ket{0}$, the resulting channel is, 
    \[\sch{F}_Z[\ch{E}](\rho) \Bigg|_{\ket{0}} \propto \chi_{_{I, I}}I\rho I + \chi_{_{Z, I}}Z\rho I + \chi_{_{I, Z}}I\rho Z + \chi_{_{Z, Z}}Z\rho Z\]
    Here diagonal elements as well as off-diagonal elements have been filtered (in this example, particularly, the $\chi_{_{X, X}}$ and $\chi_{_{Y, Y}}$ components). This brings about an increase in the fidelity of the channel (i.e. $\chi_{_{I, I}}$ \cite{wood2011tensor}), which is not possible in the Pauli-Twirling framework.
    
    \item
    The second and larger point of divergence is in the question of averaging, and in which information needs to be detectable. In the Pauli twirling case, since the goal is to remove all off-diagonal components in the Pauli transfer matrix for the noise channel, it is necessary to increase the size of the set from which the twirling operators are randomly selected as the number of qubits grows, and the size of this set increases exponentially (it is equivalent to averaging over the many possible outcomes of a set of stabiliser measurements, where the number of such measurements increases with the qubit number).
    Such an increase leads to a focus on few-qubit Clifford gates, so that the presumed independence between non-interacting qubits can do much of the work of removing the off-diagonal elements.
    We go in exactly the opposite direction: since we retain measurement information, we are able to increase the width of the Clifford circuit, and this (in conjunction with bounds on circuit depth for one- and two-qubit decompositions of general Clifford gates) helps to justify the relatively large noise assumed in the body of the circuit relative to the encoding and decoding portions.
\end{itemize}

\subsection{Iceberg Code and Entanglement-Assited QEC Codes}
The Iceberg code is a ${[[2n+2, 2n, 2]]}$ quantum error detection code defined by the stabilisers
${{X}^{\otimes 2n}}$ and ${{Z}^{\otimes 2n}}$,
which is the same distance-2 code defined by Rains to achieve the quantum Singleton bound in \cite{rains2002quantum}. These stabilisers are the same multi-qubit Pauli operators depicted in \autoref{fig:ancilla-efficient-filter}, and the latter portion of the circuit in this figure is exactly the stabiliser measurement circuit depicted in Figure~1d of \cite{self2024protecting}, but not accounting for the order of the individual two-qubit gates.

We stress however that the Iceberg code is not fault-tolerant, even though particular circuit structures may make it more robust at small scales and even though fault-tolerant approaches to stabiliser measurement may be applied. The Iceberg code is a distance-2 error detection code with a stabiliser weight that increases with the number of qubits.
Such a code cannot have a fault-tolerant QEC threshold, because increasing the code size does not increase the code distance, and the probability of an undetectable 2-qubit error therefore increases combinatorially with the code size.\\

There are two key differences between our \autoref{fig:ancilla-efficient-filter} and this distance-2 quantum error detection code:
\begin{itemize}
    \item Firstly, we adapt the encoding-portion of the circuit to any Clifford operation, whereas in the quantum error detection setting the logical operators are restricted to commute with the stabilisers.
    \item Secondly, while the Iceberg code requires an even number of physical qubits to ensure that the stabilisers commute, we do not require such a restriction in our data register
\end{itemize}.

Explaining this relaxation of the even-qubit restriction for the Iceberg code leads us to our second comparison, catalytic EAQEC.\\

Catalytic EAQEC, as described for instance in \cite{brun2014catalytic}, extends the use of QEC codes for entanglement distillation by allowing the introduction of a small number of perfect entangled pairs. This is justified on the grounds that a QEC code may encode more than one logical qubit, and so the rate can be enhanced beyond a direct use of the perfect pair for teleportation.
The role played by the entangled pair is to augment the parity check matrix, so that a QEC code can be constructed from any classical code over ${\mathbb{F}_{4}}$, without the standard additional requirement that this classical code be self-orthogonal.
Consider the parity-check matrix for the Iceberg code,
\begin{align}
    \begin{bmatrix}
        1 & 1 & 1 & \ldots & 1 \\
        \omega & \omega & \omega & \ldots & \omega
    \end{bmatrix}
    ,
\end{align}
with the standard map from ${\mathbb{F}_{4}}$ to the Pauli group under multiplication,
\begin{align}
    0 \rightarrow \hat{I},\; 1 \rightarrow \hat{X},\; \omega \rightarrow \hat{Z},\; \omega^{2} = \omega+1 \rightarrow \hat{Y}.
\end{align}
With an even number of bits this code is self-orthogonal, since the coefficients of ${1}$ and ${\omega}$ add modulo 2.
In the odd case it is not self-orthogonal, but augmenting the parity-check matrix with a perfect entangled pair we have
\begin{align}
    \begin{bmatrix}
        1 & 1 & 1 & \ldots & 1 & 1 \\
        \omega & \omega & \omega & \ldots & \omega & \omega
    \end{bmatrix}
    ,
\end{align}
where the final qubit avoids the noisy channel, and the addition of the final column ensures that all syndrome measurements commute.
Though the efficiency advantage is not apparent in this simple distance-2 code, our interest here is in the structural connection with our circuit in Figure~7 of the manuscript.\\

Suppose that we're dealing with a simple identity channel, so that in \autoref{fig:ancilla-efficient-filter}, ${Z^{\overleftarrow{\otimes} n} = Z^{\otimes n}}$ and
${X^{\overleftarrow{\otimes} n} = X^{\otimes n}}$.
Before the first controlled operations, a measurement of either ${Z^{\otimes n}}$ or ${X^{\otimes n}}$ could give any result.
The ancillary qubits both begin in ${+1}$ eigenstates of ${X}$, so that the parity check matrix at the input would be
\begin{align}
    \begin{bmatrix}
        0 & 0 & 0 & \ldots & 1 & 0 \\
        0 & 0 & 0 & \ldots & 0 & 1
    \end{bmatrix}
    .
\end{align}
Using the well-known Heisenberg-picture mapping for the action of the \textsc{cnot} gate on multi-qubit Pauli operators (where the first qubit is the control),
\begin{align*}
    X\otimes I &\rightarrow X \otimes X \quad & Z\otimes I &\rightarrow Z \otimes I \\
    I\otimes X &\rightarrow I \otimes X \quad & I\otimes Z &\rightarrow Z \otimes Z \\
\end{align*}
the first controlled-${Z^{\otimes n}}$ gate transforms the check-matrix into
\begin{align}
    \begin{bmatrix}
        \omega & \omega & \omega & \ldots & 1 & 0 \\
        0 & 0 & 0 & \ldots & 0 & 1
    \end{bmatrix}
    ,
\end{align}
and the second controlled-${X^{\otimes n}}$ then leaves us ultimately with the check-matrix
\begin{align}
    \begin{bmatrix}
        \omega & \omega & \omega & \ldots & 1 & (n\cdot\omega) \\
        1 & 1 & 1 & \ldots & 0 & 1
    \end{bmatrix}
    .
\end{align}
If the number of qubits, ${n}$, is even, then ${n\cdot\omega \cong 0}$ (taken modulo 2) and we interpret the ancillary qubits as measuring individually the stabilisers of the Iceberg code. If ${n}$ is odd, then we interpret the second ancilla in the role of a perfect qubit as in EAQEC, and the first ancilla as providing a non-destructive way to measure the first stabiliser for this `augmented Iceberg' code.\\

The paper \cite{brun2014catalytic} on catalytic EAQEC did provide a discussion on adapting the methods there for communication to the QEC, storage scenario. There are three key differences that we would like to highlight to discourage the notion that our proposal is a simple restatement or variation of the ideas in that work.

\begin{itemize}
    \item First, the method of code construction differs significantly in that this prior work builds efficient QEC codes for the noisy `identity' channel from small classical codes, while no consideration is there given to the construction of logical operations. In our proposal, the `logical' operations are defined first through the statement of the ideal Clifford channel, and the code is then derived from them.

    \item Second, the errors that can be corrected are not limited by the distance of a pre-determined classical code, and may instead be found by commuting single-qubit errors through the chosen Clifford channel, or by working backward from the errors in any measurements that might be subsequently performed, protecting specific observables.

    \item Third, as we have described above, the construction interpolates between the use of ancillary qubits as in the EAQEC case and the use of ancillary qubits solely for syndrome extraction. This allows us to define the encoding operations without reference to the data register, rather than having to consider separately for example the Iceberg and augmented-Iceberg codes for even and odd ${n}$.
\end{itemize}

To summarise, our proposal draws together a wide range of ideas that connect twirling, Flag-QEC, EAQEC, and efficient low-distance codes. There are a number of other areas in which we could draw tangential connections. The potential, given the assumptions we have made, to construct hybrid encoding procedures that link otherwise disparate subtopics in quantum error correction and mitigation, is perhaps the element of our proposal that we are most excited to explore going forward. Two other potential connections, for instance, might be a connection to operator-algebra QEC (OAQEC) \cite{dauphinais2024stabilizer} by working backwards to construct encoding circuits from errors affecting only specific observables, or a connection to error suppression via indefinite causal order using a quantum switch \cite{chiribella2021indefinite}.

\subsection{Flag-QEC and Fault-Tolerance}
\label{appendix:flag-qec}
{Flag-QEC is a fault-tolerant quantum error correction procedure designed for syndrome measurements in stabilizer codes~\cite{chao2018quantum, chamberland2018flag, chao2020flag}. It uses a small number of additional ``flag-qubits" (compared to other FTQEC schemes) to protect the syndrome qubit during the measurement procedure, so that any correlated errors (or, hook errors) occuring during this process might be detected and corrected. It does this by using a two-qubit single-basis error detection code. It differs from other common approaches in that it is very selective about encoding to detect error only at those points in the circuit where the error will be most damaging, and in that it decodes to determine the syndrome measurement information, rather than using additional ancillas to extract the information directly from the encoded state.} 

{At a first glance, the circuit structure used in this procedure (e.g, Fig.~1(c) of \cite{chao2018quantum}) shows similarities to the Pauli-$X$ filter discussed in \autoref{section:successive-filtration} of this paper. As such, a brief reading of our scheme might suggest that, Flag-QEC has all the benefits of our scheme while also providing the advantage of fault-tolerance, i.e, protection against errors generated within the scheme. We explain below that this is not the case.} \\

One key difference between our proposal and Flag-QEC, is that we assume low-error ancillas and encoding/decoding operations, precluding fault-tolerance. This assumption goes hand-in-hand however with another key difference: We construct circuits to detect arbitrary Pauli error combinations, so long as the gates in the ideal circuit are drawn from the Clifford group. This is not possible in Flag QEC, because Flag QEC relies on the existence of single-qubit $Z$-basis logical operators in the repetition code (which exist because the code does not protect against errors in this basis). This is necessary because {the flag qubit does not (and cannot, if fault-tolerance is to be preserved) interact with the data register}; there is no detectable difference between a change in relative phase induced by the data register, and one induced erroneously by a separate noisy environment.\\

In the stabilizer measurement scenario in Flag-QEC, undetected $Z$-basis errors can be tolerated because measurements can be repeated. In the more general circuit scenario this is not the case (we do not know a-priori of any repeatable measurement that would not collapse the wavefunction of the data register). As a result, it is not possible simultaneously to detect arbitrary Pauli errors and to prevent the harmful propagation of single-qubit errors in the ancillary register to multi-qubit errors in the data register. Fault-tolerance not being possible because we do not have sufficient information about the input state~\footnote{Fault-tolerant circuits exploit both symmetries in the encoded state and the non-local encoding of information. Both assumptions are absent in our scenario, so any single-qubit error propagating from an ancilla to a data qubit could be significant.}, it is necessary to bring in other assumptions about inhomogeneous error rates or circuit depth to justify an encoding operation.\\

The fault tolerance of Flag QEC depends explicitly on knowledge of the stabilizer QEC code used for the data register. As a corollary of this, in noting that Flag QEC protects against internal errors, we must acknowledge three caveats:
\begin{enumerate}
    \item fault-tolerance and robustness against backaction is a property of the QEC code rather than the ancillary state,
    \item its efficiency (in terms of the number of required flag qubits) stems solely from the fact that it protects only against single-qubit X-basis errors on the syndrome qubit, rather than errors across all data qubits,
    \item Flag QEC relies on the existence of single-qubit Z-basis logical operators, and consequently measurement information cannot be made robust.
\end{enumerate}
This stands in sharp contrast to our protocol.
With the assumption that errors in the ancilla are small relative to the bulk circuit, there is no fundamental limitation on the choice of encoding to protect the information extracted by the ancillas.

Abandoning fault-tolerance, in this setting, does not mean abandoning a reduction in the error rate (particularly when resources are scarce, as has been the motivation for a trend away from fault-tolerant circuit implementations in the recent literature, in error mitigation but also for surface code quantum computing as described e.g. in \cite{akahoshi2024partially}). This we have argued in Section~V of the manuscript. It must however be possible to argue, as for virtual distillation in error mitigation, and as for entanglement distribution in quantum communication, that the additional overheads contribute only a small amount of noise relative to the initial noisy quantum channel.

\section{Robustness to internal noise}
\label{sec:internal-noise}

Even though our Pauli Filter scheme cannot achieve fault-tolerance due to the assumption of noiseless (or low-error ancillas), it can ensure robustness to internal errors in noise-regimes in which fault-tolerant schemes such as Flag-QEC perform worse. This section provides numerical evidence for this robustness to internal noise. We use the Flag-QEC scheme as a point of reference, and compare both these protocols by numerical simulations to probe regimes where they are advantageous. We show evidence that, even if errors are generated within our protocol, there exists a practical noise-regime in which our scheme exhibits superior error-reduction compared to Flag-QEC. Due to the simplicity of the circuit structure, noise generated within our scheme can remain low and not drastically impact the output error-rate.\\ 

Fault-tolerant implementations are characterized by the existence of a physical noise-threshold, below which expansion of the code by recursive encoding or concatenation, can arbitrarily decrease the logical error-rate. To gain the advantage of a fault-tolerant construction, one therefore needs sub-threshold physical error-rates and enough resources to expand the encoding.
If the error-rates are above the threshold, fault-tolerant schemes become impractical.
It is this noise-regime we want to stress, and within which we find a numerical advantage.
Thus, although our construction is not fault-tolerant, it can lower error-rates by a notable margin on very noisy devices with limited resources.\\

\paragraph{Simulation Setup}
Here, we provide details of the setup for noisy simulations comparing the Flag-QEC protocol and the Pauli-Filter protocol.
In the interest of fairness, we choose conditions that would result in the least overhead for Flag-QEC.
On the left in Figure \ref{fig:noisy_simulation_1qubit}, we encode 1 qubit in the [[5, 1, 3]] code (needing 4 ancillary qubits), which then undergoes an isotropic noise channel across all the qubits, followed by the Flag-QEC protocol to correct the errors, as outlined in \cite{chao2018quantum} (this requires 2 further ancillas).
Since the [[5, 1, 3]] code is the smallest 1-qubit error-correction code, this gives the smallest qubit overhead (6 extra qubits) and gate overhead to implement Flag-QEC.
On the right, we implement the Pauli-Filter protocol on 1 qubit, as outlined in Fig. \ref{fig:noisy_simulation_1qubit} of our paper. We also extend this simulation to 2 encoded qubits (as outlined in Fig. \ref{fig:noisy_simulation_2qubit}), to compare the scaling behaviour of the two protocols. 
\begin{figure}[!h]
    \resizebox{\columnwidth}{!}{
\begin{yquantgroup}
    \registers{
    qubit {} q; 
    qubit {} code; 
    qubit {} flag;
    }
    \circuit{
    init {$\ket{\psi}$} q;
    discard code, flag;
    
    hspace {0.5 cm} q;
    
    init {$\ket{0}^{\otimes 4}$} code;
    slash code;
    
    box {Encode in\\ {[[5, 1, 3]]}} (q, code);
    hspace {0.25 cm} q, code, flag;

    init {$\ket{0}^{\otimes 2}$} flag;
    slash flag;
    
    [this subcircuit box style={draw, dashed, rounded corners, fill=red!20, inner ysep=6pt, inner xsep=0pt}, "noise" above]
        subcircuit {qubit {} q; 
        hspace {1.5 cm} q;
        } q;

    [this subcircuit box style={draw, dashed, rounded corners, fill=red!20, inner ysep=6pt, inner xsep=0pt}, ]
        subcircuit {qubit {} q; 
        hspace {1.5 cm} q;
        } code;
        
    hspace {0.25 cm} q, code, flag;
    
    
    box {Flag\\QEC} (q, code, flag);
    discard flag;
    output {Output} (q, code);
    }
    \equals[]
    \circuit{
    init {$\ket{\psi}$} q;
    discard code, flag;
    hspace {0.5 cm} q;
    init {$\ket{0}$} code;
    init {$\ket{0}$} flag;
    h code, flag; 
    box {$X$} q | flag;
    box {$Z$} q | code;
    hspace {0.25 cm} q, code, flag;
    
    [this subcircuit box style={draw, dashed, rounded corners, fill=red!20, inner ysep=6pt, inner xsep=0pt},
    "noise" above]
    subcircuit {qubit {} q; 
    hspace {1.5 cm} q;
    } q;

    [this subcircuit box style={draw, dashed, rounded corners,
    fill=red!20, inner ysep=6pt, inner xsep=0pt}]
    subcircuit {qubit {} q; 
    hspace {1.5 cm} q;
    } code, flag;
    
    hspace {0.25 cm} q, code;
    box {$Z$} q | code;  
    box {$X$} q | flag;
    h code, flag; 
    
    measure code, flag;
    box {$X$} q| code;
    box {$Z$} q| flag;
    discard code, flag;
    hspace {0.5 cm} q;
    }
\end{yquantgroup}
}
\caption{Noisy simulation of Flag-QEC (left) and Pauli-Filter (right), for 1 encoded qubit.}
\label{fig:noisy_simulation_1qubit}
\end{figure}

\begin{figure}[!h]
    \resizebox{\columnwidth}{!}{
\begin{yquantgroup}
    \registers{
    qubit {} flag1; 
    qubit {} code1; 
    qubit {} q1;
    qubit {} q2;
    qubit {} code2; 
    qubit {} flag2;
    }
    \circuit{
    init {$\ket{\Psi}$} (q1, q2);
    discard code1, code2, flag1, flag2;
    
    hspace {0.5 cm} q1, q2;
    
    init {$\ket{0}^{\otimes 4}$} code1, code2;
    slash code1, code2;
    
    box {Encode in\\ {[[5, 1, 3]]}} (q1, code1);
    box {Encode in\\ {[[5, 1, 3]]}} (q2, code2);
    
    hspace {0.25 cm} q1, q2, code1, code2, flag1, flag2;

    init {$\ket{0}^{\otimes 2}$} flag1, flag2;
    slash flag1, flag2;

    [this subcircuit box style={draw, dashed, rounded corners, fill=red!20, inner ysep=6pt, inner xsep=0pt}, ]
        subcircuit {qubit {} q; 
        hspace {1.5 cm} q;
        } q1, q2, code1, code2;
        
    hspace {0.25 cm} q1, q2, code1, code2, flag1, flag2;
    
    box {Flag\\QEC} (q1, code1, flag1), (q2, code2, flag2);
    discard flag1, flag2;
    output {Output} (q1, code1, q2, code2);
    }
    \equals[]
    \circuit{
    init {$\ket{\Psi}$} (q1, q2);
    discard code1, code2, flag1, flag2;
    hspace {0.5 cm} q1, q2;
    
    init {$\ket{0}$} code1, code2, flag1, flag2;
    h code1, code2, flag1, flag2; 
    box {$X$} q1 | flag1;
    box {$Z$} q1 | code1;
    box {$X$} q2 | flag2;
    box {$Z$} q2 | code2;
    hspace {0.25 cm} q1, code1, flag1, q2, code2, flag2;

    [this subcircuit box style={draw, dashed, rounded corners,
    fill=red!20, inner ysep=6pt, inner xsep=0pt}]
    subcircuit {qubit {} q; 
    hspace {1.5 cm} q;
    } q1, q2, flag1, flag2, code1, code2;
    
    hspace {0.25 cm} q1, q2;
    box {$Z$} q1 | code1;  
    box {$X$} q1 | flag1;
    box {$Z$} q2 | code2;  
    box {$X$} q2 | flag2;
    h code1, code2, flag1, flag2; 
    
    measure code1, flag1, code2, flag2;
    box {$X$} q1| code1;
    box {$Z$} q1| flag1;
    box {$X$} q2| code2;
    box {$Z$} q2| flag2;
    discard code1, code2, flag1, flag2;
    hspace {0.5 cm} q1;
    }
\end{yquantgroup}
}
\caption{Noisy simulation of Flag-QEC (left) and Pauli-Filter (right), for 2 encoded qubits.}
\label{fig:noisy_simulation_2qubit}
\end{figure}
In order to analyze the performance drop due to internal noise, we subject all qubits in both schemes to depolarizing noise-channels of equal strengths $p_{\rm channel}$, i.e, \textbf{we drop the assumption of noiseless ancillas in the Pauli-Filter in these simulations}.
Additionally, \textbf{we also drop the assumption of perfect gates} in the simulations for both these schemes. All the two-qubit gates in the Encoding and Flag-QEC circuits on the left and in the Filter operations on the right are assumed to be noisy: each two-qubit gate is followed by a local depolarizing channel of strength $p_{\rm gate}$ on both the control and the target qubit.
This setup has the most detrimental effect on our Pauli-Filter construction, since we discard all the previous assumptions we had made. In both scenarios, the correction operations are noiseless. 

We perform simulations of these two scenarios for $p_{\rm channel}$ ranging from 0 to 0.1 and $p_{\rm gate}$ ranging from 0 to 0.01. Finally, we report the fidelity of the output state to the ideal noiseless state, averaged over 10,000 runs. A plot of the fidelities for all combinations of ($p_{\rm channel}, p_{\rm gate}$) is provided in Fig.~\ref{fig:flag_vs_filter_line}. \\

\begin{figure}[h!]
\centering
\includegraphics[width=1\columnwidth]{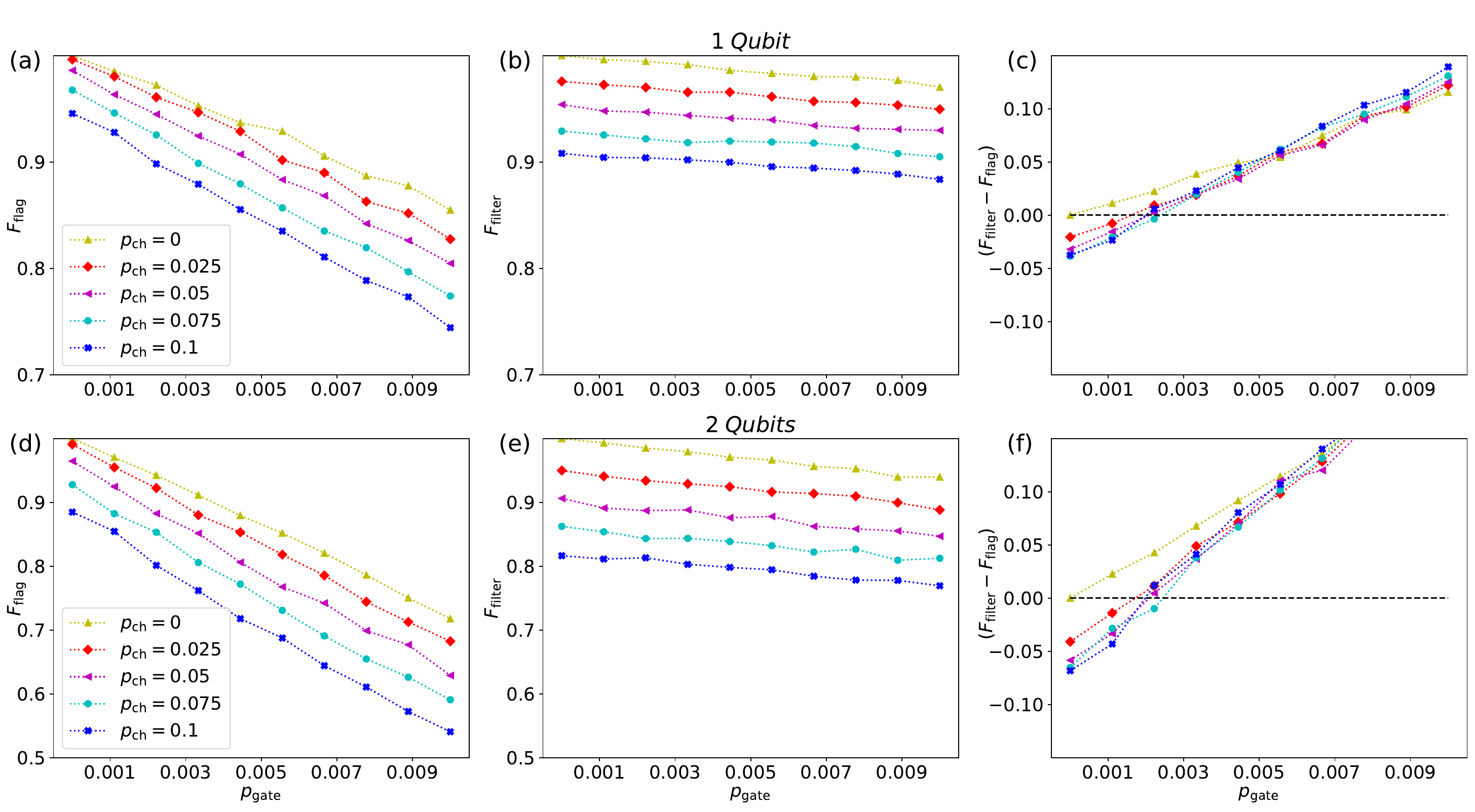}
\caption{ 
 Average fidelity of the output state for the Flag-QEC scheme applied to the [[5, 1, 3]] code (left), the Pauli Filter scheme (middle), and the difference in these fidelities (right), for 1 encoded logical qubit (top), and 2 encoded logical qubits (bottom).  
}
\label{fig:flag_vs_filter_line}
\end{figure}

\paragraph{Result Analysis} From the first two plots in Fig. \ref{fig:flag_vs_filter_line}, one can see that the performance of the Pauli-Filter is more sensitive to $p_{\rm channel}$, whereas, that of Flag-QEC is more sensitive to $p_{\rm channel}/p_{\rm gate}$. The third plot gives the difference in performance of the two schemes, and, is the most illuminating. One notice that Flag-QEC performs better if $p_{\rm gate}$ is much lower than $p_{\rm channel}$, while the Pauli-Filter performs better if they are comparable. For example, for the 1-qubit case, if $p_{\rm channel} = 0.1, p_{\rm gate} = 0.001$, the average fidelity of the output state of Flag-QEC is higher than that of the Pauli-Filter by a margin of $\approx 0.04$. On the other hand, for the same channel error-rate of $p_{\rm channel} = 0.1$, if the gate error-rate is significantly higher, e.g, $p_{\rm gate} = 0.01$, the Pauli Filter outperforms Flag-QEC by a margin of $\approx 0.14$. The break-even point seems to be around $p_{\rm gate} = 0.002$, above which the Pauli-Filter performs better. This trend is even sharper for the 2-qubit case (Fig. \ref{fig:flag_vs_filter_line} bottom). In this case, the output fidelity of Flag-QEC falls off even more rapidly compared to the Pauli-filter, for $p_{\rm gate}$ is greater than the break-even value of 0.002. \\

The results shown in Fig.~\ref{fig:flag_vs_filter_line} are expected: even in this favorable case the gate overhead of Flag-QEC is much higher than that of the Pauli-Filter, and so for moderate gate-noise far from its noise-threshold, the performance of Flag-QEC declines sharply.
For larger codes, more gates and more qubits will be required to perform Flag-QEC, which will result in an even sharper decrease in performance with increasing gate-noise.
Thus, in this far-from-threshold low-resource regime precluding fault-tolerance, a situation typical in current hardware, we expect the Pauli-Filter to have the advantage.

\color{black}



\newcommand{\CDC}[1]{Proceedings of the #1 IEEE Conference on Decision and
  Control}

\end{document}